\documentclass[12pt,letterpaper]{article}

\usepackage{natbib,hyperref}

\usepackage[ left=1in, top=1in, right=1in, bottom=1in]{geometry}
\usepackage{graphicx,bm,colonequals,amsmath,amssymb,url,xcolor,bbm}
\usepackage{array,tabularx,multirow}
\usepackage{enumitem}
\usepackage[font={footnotesize}]{caption,subcaption}
\usepackage[utf8]{inputenc}
\usepackage{enumitem}
\usepackage{soul} 
\usepackage{placeins} 
\usepackage[normalem]{ulem}
\usepackage{float}

\usepackage{graphbox}

\graphicspath{{plots/}}

\usepackage{mathtools}
\mathtoolsset{showonlyrefs} 

\bibpunct[, ]{(}{)}{;}{a}{,}{,}
   
\usepackage{amsthm}
\newtheoremstyle{propstyle} 
    {2mm}                    
    {1mm}                    
    {\itshape}                   
    {}                           
    {\scshape}                   
    {.}                          
    {.5em}                       
    {}  
\theoremstyle{propstyle}
\newtheorem{proposition}{Proposition}
\theoremstyle{propstyle}
\newtheorem{definition}{Definition}
\theoremstyle{propstyle}
\newtheorem{lemma}{Lemma}
\theoremstyle{propstyle}

\theoremstyle{propstyle}

\usepackage{array,framed}
\newcounter{algorithm}

\makeatletter
\renewcommand{\paragraph}{%
  \@startsection{paragraph}{4}%
  {\z@}{2ex \@plus 1ex \@minus .2ex}{-1em}%
  {\normalfont\normalsize\bfseries}%
}
\makeatother

\DeclareMathAlphabet\mathbfcal{OMS}{cmsy}{b}{n}

\newcommand{\bs}{\mathbf{s}}

\newcommand{\bx}{\mathbf{x}}
\newcommand{\by}{\mathbf{y}}

\newcommand{\bg}{\mathbf{g}}

\newcommand{\bz}{\mathbf{z}}

\newcommand{\bA}{\mathbf{A}}

\newcommand{\bI}{\mathbf{I}}
\newcommand{\bD}{\mathbf{D}}

\newcommand{\bU}{\mathbf{U}}

\newcommand{\bK}{\mathbf{K}}

\newcommand{\bM}{\mathbf{M}}

\newcommand{\bfzero}{\mathbf{0}}

\newcommand{\bfmu}{\bm{\mu}}

\newcommand{\bftheta}{\bm{\theta}}

\newcommand{\bfepsilon}{\bm{\epsilon}}

\newcommand{\bfSigma}{\bm{\Sigma}}

\DeclareMathOperator*{\cov}{cov}
\newcommand{\diag}{diag}

\newcommand{\GP}{\mathcal{GP}}

\newcommand{\order}{\mathcal{O}}
\newcommand{\normal}{\mathcal{N}}

\newcommand{\inputspace}{\mathcal{X}}

\DeclareMathOperator*{\argmax}{arg\,max}

\title{Correlation-based sparse inverse Cholesky factorization for fast Gaussian-process inference}

\author{Myeongjong Kang\thanks{Department of Statistics, Texas A\&M University} \and Matthias Katzfuss\footnotemark[1] \thanks{Corresponding author: \texttt{katzfuss@gmail.com}}}

\date{}

\begin{document}

\maketitle

\begin{abstract}
\noindent
Gaussian processes are widely used as priors for unknown functions in statistics and machine learning. To achieve computationally feasible inference for large datasets, a popular approach is the Vecchia approximation, which is an ordered conditional approximation of the data vector that implies a sparse Cholesky factor of the precision matrix. The ordering and sparsity pattern are typically determined based on Euclidean distance of the inputs or locations corresponding to the data points. Here, we propose instead to use a correlation-based distance metric, which implicitly applies the Vecchia approximation in a suitable transformed input space. The correlation-based algorithm can be carried out in quasilinear time in the size of the dataset, and so it can be applied even for iterative inference on unknown parameters in the correlation structure. The correlation-based approach has two advantages for complex settings: It can result in more accurate approximations, and it offers a simple, automatic strategy that can be applied to any covariance, even when Euclidean distance is not applicable. We demonstrate these advantages in several settings, including anisotropic, nonstationary, multivariate, and spatio-temporal processes. We also illustrate our method on multivariate spatio-temporal temperature fields produced by a regional climate model.
\end{abstract}



\section{Introduction \label{sec:intro}}

Gaussian processes (GPs) are used for modeling functions in a variety of settings, including geostatistics \citep[e.g.,][]{Banerjee2004,Cressie2011}, nonparametric regression and machine learning \citep[e.g.,][]{Rasmussen2006}, the analysis of computer experiments \citep[e.g.,][]{Sacks1989, Kennedy2001, gu2018scaled}, and optimization \citep{Jones1998}. GPs can also be used to represent wide neural networks \citep{Yang2019}.
However, direct application of GPs requires working with and decomposing the data covariance matrix at a cost that is cubic in the data size, which is often too expensive for today's large datasets.

Many approaches have been proposed to scale GP inference to large numbers of observations \citep[see][for recent reviews]{Heaton2017,Liu2018}.
Among these, probably the most promising class of approximations in spatial statistics consists of the Vecchia approximation \citep{Vecchia1988} and its extensions \citep[e.g.,][]{Stein2004,Datta2016,Guinness2016a,Sun2016,Katzfuss2017a,Katzfuss2018,Schafer2020}. As detailed in \citet{Katzfuss2017a}, the class also contains many other popular GP approximations as special cases \citep[e.g.,][]{Snelson2007,Finley2009,Sang2011a,Eidsvik2012a,Datta2016,Katzfuss2015,Katzfuss2017b} and it is closely related to composite-likelihood methods \citep[e.g.,][]{Varin2008,Eidsvik2012}. 
Vecchia approximations obtain a sparse Cholesky factor of the precision matrix via an ordered conditional approximation, based on removing conditioning variables in a factorization of the joint density of the GP observations into a product of conditional distributions.

The performance of a Vecchia approximation depends heavily on the choice of ordering of the variables and the choice of conditioning sets (which determines the Cholesky sparsity pattern). So far, Vecchia approximations have been mostly applied in geospatial applications featuring isotropic GPs in low-dimensional input spaces, for which the ordering and conditioning can be carried out based on the inputs or locations. Specifically, the observations can be ordered using a maximum-minimum-distance algorithm, and the sparsity is determined by nearest-neighbor conditioning \citep{Guinness2016a}. Both ordering and conditioning are typically carried out based on Euclidean distance of the corresponding inputs. We call this existing approach Euclidean-based Vecchia (EVecchia).
EVecchia has also been used for nonisotropic settings, including for nonstationary \citep{Konomi2019,Risser2020}, multivariate \citep{zhang2021high}, space-time \citep[]{white2019nonseparable}, and periodic GPs \citep[][Supplement A.9]{Datta2016}.

Here, we propose Vecchia approximations whose ordering and conditioning employ a correlation-based distance metric; we refer to this approach as correlation-based Vecchia or CVecchia. Correlation-based conditioning (but not ordering) was already mentioned in the early Vecchia papers \citep{Vecchia1988,Jones1997,Stein2004}, but it was dismissed and not thoroughly explored, mainly due to concerns about high computational cost and instability. In contrast, we argue that CVecchia can improve approximation accuracy, and it can be carried out efficiently even in the presence of unknown parameters, allowing both frequentist and Bayesian parameter inference. 
\citet{yu2017geometry} proposed a related correlation-based idea in the context of hierarchical low-rank compression (but not factorization) of a positive-definite matrix.
So far, all previous Vecchia approaches have based the ordering on spatial or temporal locations, without considering the covariance function to be approximated. Conditioning sets have also been selected based on the locations; one exception is the dynamic spatio-temporal nearest-neighbor GP \citep{Datta2016a}, whose adaptive neighbor-selection scheme defines a space-time distance as a function of the spatio-temporal covariance.

EVecchia and CVecchia are equivalent for strictly decreasing isotropic correlation functions \citep{Jones1997,Stein2004}, but CVecchia has two advantages for more complex situations, such as anisotropic, nonstationary, multivariate, and spatio-temporal processes: It can provide much higher accuracy, and it offers a simple, automatic strategy even when Euclidean distance is not applicable. Thus, CVecchia greatly expands the applicability of the Vecchia approach; in fact, CVecchia can be applied to any covariance matrix whose individual entries can be obtained or computed quickly, as the approximation only relies on evaluating or accessing a near-linear number of entries. CVecchia implicitly applies a Vecchia approximation in a suitable transformed input domain, in which the GP of interest is isotropic and Euclidean distance is meaningful.

The remainder of this document is organized as follows. In Section~\ref{sec:stanV}, we review Vecchia approximations from a perspective that enables our extensions. In Section~\ref{sec:corrV}, we introduce correlation-based Vecchia and discuss its properties. Section~\ref{sec:simul} provides numerical comparisons. In Section~\ref{sec:realdat}, we illustrate the performance of our method using output from a regional climate model. Section~\ref{sec:conc} concludes and discusses future work. Appendix~\ref{app:proofs} contains proofs. 
The code for running our method and reproducing figures can be found at \url{https://github.com/katzfuss-group/correlationVecchia}.


\section{Review of Euclidean-based Vecchia \label{sec:stanV}}

\subsection{The Vecchia approximation \label{sec:vecchia}}

Consider a centered Gaussian random vector $\by = \left( y_1 , y_2 , \ldots , y_n \right)^{\top} \sim \normal_n(\bfzero, \bK)$,
where $\bK$ is an $n \times n$ positive-definite covariance matrix.
For example, $\by$ may be a vector of observations of a GP.
Evaluating the Gaussian density $p(\by)$, which typically relies on Cholesky decomposition of $\bK$, generally requires $\order (n^3)$ computing time and $\order (n^2)$ memory; this is often too expensive for large $n \gg 10^3$.

A promising approach to reduce the computational effort is the Vecchia approximation. Motivated by the exact factorization $p(\by) = \prod_{i=1}^n p(y_i | \by_{1:i-1})$ with $\by_{1:0} \colonequals \emptyset$, the Vecchia approximation is given by
\begin{equation}
\label{eq:vecchia}
\hat p(\by) = \prod_{i=1}^n p(y_i | \by_{c(i)}) = \normal_n(\bfzero,\hat\bK),
\end{equation}
where $c(1) = \emptyset$ and $c(i) \subset \lbrace 1, \ldots, i-1 \rbrace$ for $i = 2, \ldots, n$. We assume that all conditioning sets are at most of size $m$, $|c(i)| = \min( m, i-1)$, for some integer $m\ll n$. The approximate covariance matrix $\hat{\bK}$ has a sparse inverse Cholesky factor: $\hat{\bK}^{-1} = \bU\bU^{\top}$, where $\bU$ is a sparse upper triangular matrix with at most $m$ off-diagonal nonzeros per column, given by $\bU_{\tilde c(i),i} = (\bK_{\tilde c(i),\tilde c(i)})^{-1}\mathbf{e}_1/\big(\mathbf{e}_1^\top(\bK_{\tilde c(i),\tilde c(i)})^{-1}\mathbf{e}_1\big)^{1/2}$, where $\tilde c(i) = \{i\} \cup c(i)$ and $\mathbf{e}_1$ is a vector of length $m+1$ with the first entry equal to one and all other entries equal to zero \citep{Schafer2020}.
Each of the $n$ columns of $\bU$ can be computed in $\mathcal{O}(m^3)$ time, completely in parallel.
Further, the $\bU$ implied by the Vecchia approximation is the optimal sparse inverse Cholesky factor of $\bK$ in terms of Kullback-Leibler (KL) divergence between $\mathcal{N}(\mathbf{0}, \bK)$ and $\mathcal{N}(\mathbf{0},(\bU\bU^{\top})^{-1})$ for the sparsity pattern for $\bU$ implied by the $c(i)$ as above \citep{Schafer2020}.

The size of conditioning sets, $m$, acts as a tuning parameter that trades off sparsity and computational speed against approximation accuracy.
In particular, if $m = 0$, then the Vecchia approximation assumes diagonal $\hat{\bK}$ and yields independent $y_1,\ldots,y_n$. If $c(i) = \{1,\ldots,i-1\}$ and hence $m = n-1$, then the Vecchia approximation is exact. In general, adding indices to the conditioning sets is guaranteed to result in lower or equal KL divergence \citep{Guinness2016a}. In many settings, high accuracy can be achieved even using relatively small $m$. In practice, often $m<100$ is chosen with respect to available computational resources \citep[see, e.g., the guidelines and discussion in][]{Katzfuss2017a}.

\subsection{Ordering and conditioning \label{sec:ordcond}}

For given $m$, the accuracy of a Vecchia approximation depends on the choice of ordering of the variables  $y_1,\ldots,y_n$ in $\by$, and on the choice of conditioning sets $c(m+2), \ldots , c(n)$. 
Arguably the preferred approach in this setting is to combine a maximum-minimum-distance ordering \citep[MM;][]{Guinness2016a} and nearest-neighbor conditioning (NN), as illustrated in Figure \ref{fig:mmnn}.

Specifically, for MM ordering, the first index $i_1$ can be selected arbitrarily (e.g., $i_1=1$), and then the subsequent indices are selected for $k = 2, \ldots , n$ as
\begin{equation}
\label{eq:mm}
i_k = \argmax_{i \, \in \, \mathcal{I} \setminus \mathcal{I}_{1:k-1}} \,\, \min_{j \, \in \, \mathcal{I}_{1:k-1}} \tau(i,j) ,
\end{equation}
where $\mathcal{I} = \{1,\ldots,n\}$ and $\mathcal{I}_{1:k-1} = \{i_1 , \ldots , i_{k-1}\}$, using a predefined distance measure $\tau$ between the entries of $\by$. For simplicity of notation, assume henceforth and in \eqref{eq:vecchia} that $\by=(y_1,\ldots,y_n)$ follows MM ordering (i.e., $y_k = y_{i_k}$).

For NN conditioning, $y_{i}$ conditions on the $\min ( m, {i}-1)$ previously ordered variables $\by_{c(i)}$ that are nearest to $y_{i}$ in terms of $\tau$. Specifically, for $1 < i \leq m+1$, we have $c(i) =\{1,\ldots,i-1\}$; for $i>m+1$, we have 
\begin{equation}
\label{eq:nn}
   c(i) \subset \{1,\ldots,i-1\} \text{ of size } |c(i)|=m, \text{ s.t.\ } \tau(i,j) \leq \tau(i,k) \,  \forall j \in c(i), \; k \in \{1,\ldots,i-1\}\setminus c(i). 
\end{equation}
We also employ an algorithm that groups similar conditioning sets \citep{Guinness2016a} to lessen overall computational cost of Vecchia approximation.
Although we only consider conditioning sets consisting of the $m$ nearest neighbors here, our framework also allows the use of other neighbor-selection strategies. For instance, \citet{Schafer2020} uses conditioning sets consisting of all variables within a ball of a certain radius, which decreases systematically with the MM-ordering index $i$; however, we carried out exploratory numerical studies, in which this radius-based approach was often significantly less accurate than NN conditioning, especially for irregularly spaced inputs.

As we can see, specifying a Vecchia approximation requires a choice of distance $\tau(i,j)$ between pairs $\left( y_i, y_j \right)$ to determine MM and NN.
So far, the Vecchia approximation has been applied in the setting where $\by$ is a realization of a GP $y(\cdot) \sim \GP(0,K)$ at inputs $\bx_1,\ldots,\bx_n$, so that $y_i = y(\bx_i)$ and $\bK_{ij} = K(\bx_i,\bx_j)$. Then, the ordering and conditioning for $y_1,\ldots,y_n$ are typically based on the Euclidean distance between corresponding inputs:
\begin{equation}
\tau_E(i,j) = \| \bx_i - \bx_j \|,
\end{equation}
which we call Euclidean-based maximum-minimum-distance ordering (E-MM) and Euclidean-based nearest neighbor conditioning (E-NN), respectively. E-MM and E-NN are illustrated in Figure~\ref{fig:mmnn}. We refer to a Vecchia approximation based on this approach as EVecchia (which is then only a function of $m$).
EVecchia has been shown to outperform Vecchia approximations based on other ordering and conditioning schemes for GPs in low-dimensional input spaces \citep[e.g.,][]{Guinness2016a,Katzfuss2017a,Schafer2020}.

\begin{figure}
    \centering
    \includegraphics[width=.95\linewidth]{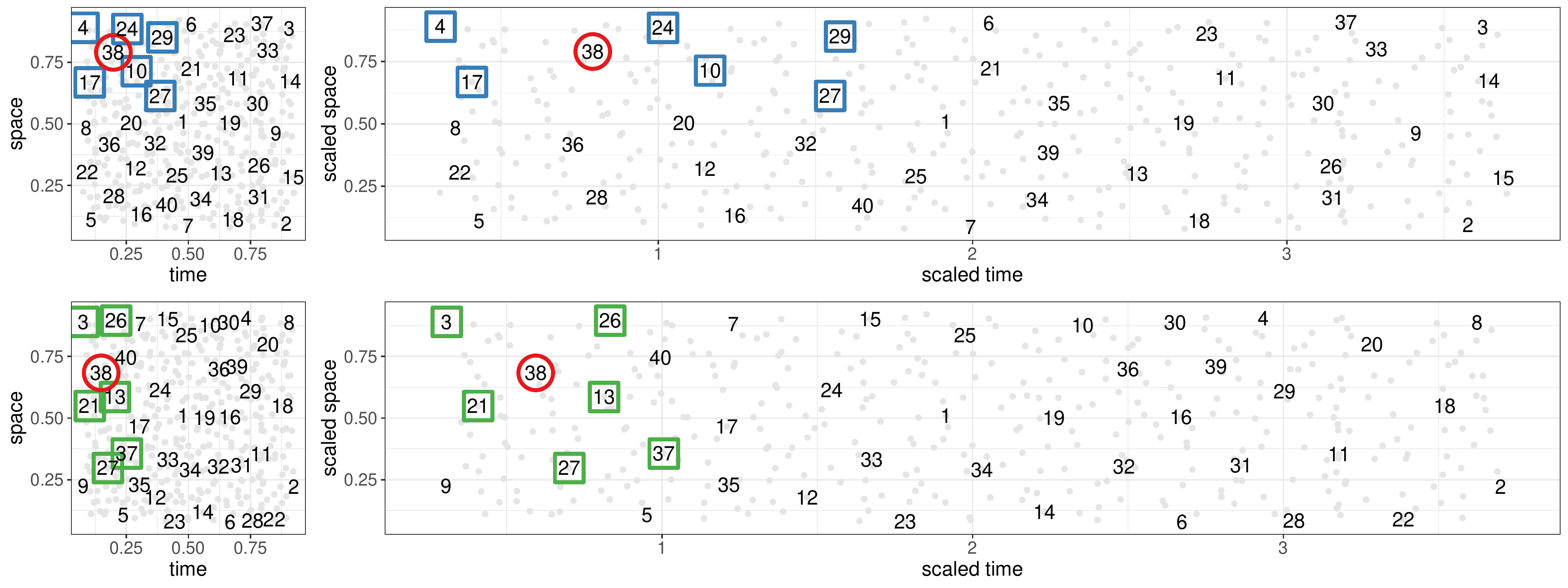}
    \caption{Euclidean (top) and correlation-based (bottom) maximum-minimum-distance ordering (MM) and nearest-neighbor conditioning (NN) for $n=400$ spatio-temporal inputs (grey points), assuming a spatio-temporal covariance \eqref{eq:spacetime} with the ratio of temporal to spatial range, $r_t/r_s=4$. The first 40 ordered inputs are numbered, and boxes denote the nearest $m=6$ previously ordered neighbors $c(i)$ of the $i=38$th input (red circle), in the unit-square input space $[0,1]^2$ (left panels) and the transformed input space $[0,4] \times [0,1]$ (right). Correlation-based MM and NN can be thought of as Euclidean MM and NN in the transformed input space (bottom right). This figure is inspired by Figure 1 in \cite{Katzfuss2020}.}
    \label{fig:mmnn}
\end{figure}

\section{Correlation-based Vecchia approximation \label{sec:corrV}}

\subsection{Definition and overview \label{sec:corrdef}}

We propose a correlation-based Vecchia (CVecchia) approximation of $\by \sim \normal_n(\bfzero,\bK)$. CVecchia consists of a Vecchia approximation \eqref{eq:vecchia} for which the MM ordering \eqref{eq:mm} and NN conditioning \eqref{eq:nn} are carried out using a correlation-based distance,
\begin{equation}
\label{eq:corrdist}
    \tau_C(i,j) = (1-|\rho_{ij}|)^{1/2}, \qquad \text{where} \;  \rho_{ij} = \bK_{ij}/(\bK_{ii}\bK_{jj})^{1/2}, \qquad i,j \in \mathcal{I} = \{1,\ldots,n\}
\end{equation}
which we will call C-MM and C-NN, respectively. 

As we will explore in more detail below, CVecchia is equivalent to EVecchia for many popular isotropic kernels; CVecchia can be more accurate than EVecchia for nonisotropic kernels (e.g., anisotropic, nonstationary, spatio-temporal); and CVecchia is applicable even when EVecchia is not (e.g., multivariate GPs, GPs based on discrete or non-Euclidean inputs such as in text analysis or natural language processing).

Provided that $\bK$ is positive-definite, $\tau_C: \mathcal{I} \times \mathcal{I} \rightarrow [0,1]$ in \eqref{eq:corrdist} is a proper distance metric \citep{van2012metric} and, in particular, satisfies the triangle inequality. 
This allows us to rapidly compute C-MM and C-NN with an adaptation of the algorithm in \citet{Schafer2020} in quasilinear time in $n$, assuming that each entry of $\bK$ can be computed in $\order(1)$ time; in practice, this computational cost is often small relative to that of the core Vecchia approximation in \eqref{eq:vecchia}, and so the computational complexity of CVecchia can still be thought of as $\order(nm^3)$, same as for EVecchia. Among other things, this means that CVecchia is useful even when $\bK$ depends on unknown parameters that must be inferred.

As MM and NN only depend on the ranking of distances (not the distances themselves), 
other correlation-based distance metrics that are ordinally equivalent \citep[e.g.,][]{van2012metric} to $\tau_C$ in \eqref{eq:corrdist} will result in equivalent CVecchia approximations.

By definition of the correlation-based distance in \eqref{eq:corrdist}, C-MM and C-NN ignore the marginal variances of the variables $y_1,\ldots,y_n$. Thus, one may ask whether, for example, a better conditioning set $c(i)$ could be obtained based on a distance metric that takes into account highly varying marginal variances. However, this is not the case. To see this, note that we have $\text{KL}(p(\by) \vert \hat p(\by)) = \sum_{i}\log(var(y_i | \by_{c(i)})/var(y_i | \by_{1:i-1}))/2$ \citep[e.g.,][]{Guinness2016a}, where $var(y_i|a y_j)$ is the same for any $a \neq 0$ and so $var(y_i | \by_{c(i)})$ does not depend on the marginal variances of the conditioning variables.

\subsection{Properties of CVecchia in the special case of reducible GPs \label{sec:mathprop}}

CVecchia is equivalent to EVecchia if $\bK$ is the covariance matrix of a realization of an isotropic GP with strictly decreasing positive covariance function: $\bK_{ij} = K(\bx_i,\bx_j) = \sigma^2 \rho(\tau_E(i,j))$, where $\rho: \mathbb{R}^+_0 \rightarrow [0,1]$ is strictly decreasing; examples include Mat\'ern and power-exponential covariance functions.
Taken one step further, this finding suggests that CVecchia can be interpreted as EVecchia on a transformed input space in the special case of \emph{reducible} GPs, which we define as follows:
\begin{definition}[$q$-reducibility]\label{def:reducibility}
A zero-mean Gaussian process $y(\cdot)$ on $\mathbb{R}^d$ with $d \ge 1$ is $q$-reducible if there exists a $\psi: \mathbb{R}^d \rightarrow \mathbb{R}^q$ such that $y\big( \psi^{-1}(\cdot) \big)$ is a Gaussian process with a strictly decreasing isotropic covariance function. In particular, $y$ is bijectively reducible if $q=d$.
\end{definition}
Definition \ref{def:reducibility} is broad enough to include many GPs of interest. For some covariance functions, the deformation function $\psi$ can be easily identified, including (geometrically) anisotropic GPs, automatic relevance determination, and latent-dimension (i.e., dimension-expansion) approaches to multivariate and spatio-temporal GPs. Also, some popular nonstationary GPs are explicitly constructed in the way we define the reducibility \citep[e.g.,][]{perrin1999modelling, Schmidt2003, vu2020modeling}.

A major advantage of CVecchia is that it is not required to identify the deformation $\psi$ explicitly, but that it automatically carries out the approximation in a transformed space in which Euclidean distance is meaningful:
\begin{proposition}\label{prop:equivalence}
Assume that a zero-mean Gaussian process $y(\cdot)$ is $q$-reducible with respect to $\psi$. If the first index is chosen to be the same for both C-MM and E-MM, then CVecchia of $y(\cdot)$ at inputs $\bx_1 , \ldots , \bx_n$ is identical to EVecchia of $y\big( \psi^{-1}(\cdot) \big)$ at the transformed inputs $\psi(\bx_1), \ldots ,\psi(\bx_n)$.
\end{proposition}

The dimension $q$ in Proposition \ref{prop:equivalence} is important, in that EVecchia approximations become more challenging as the input dimension increases. There have been studies on necessary and sufficient conditions for reducibility and how large $q$ must be \citep[e.g.,][]{perrin2000reducing,curriero2006use}, and sufficient conditions for related concepts have been identified \citep[e.g.,][]{porcu2010non,perrin2003nonstationarity,perrin2007can}. In some settings, theoretical guarantees depending on $q$ on the performance of CVecchia for reducible GPs can be provided using recent results for isotropic GPs \citep{Schafer2020}. For example, if a process is $q$-reducible to an isotropic GP whose kernel is the Green's function of an elliptic PDE (which is equivalent to a Mat\'ern covariance up to edge effects), then CVecchia can provide an $\epsilon$-accurate approximation in $\order\left( n \log^{3q}(n/\epsilon) \right)$ time.

While Proposition~\ref{prop:equivalence} provides an explanation for why CVecchia can produce adaptive approximations to some popular nonisotropic GPs, this property deserves further investigation for its relationship to Euclidean embeddings \citep{witsenhausen1986minimum,matousek2013lectures, maehara2013euclidean}. It is well-known that an exact representation of a given metric space into Euclidean space is not easy to find and that is why approximate embeddings have been studied. For instance, the Johnson-Lindenstrauss flattening lemma \citep{johnson1984extensions} states the existence of low-distortion (no more than a factor of $1 \pm \epsilon$) Euclidean embedding of a given finite metric space to $q$-dimensional Euclidean space where $q \ge \order \left( \log (n) / \epsilon \right)$. This may provide another way to carry out performance evaluation of CVecchia approximations.

\subsection{Estimation of parameters \label{sec:estim}}

So far, we have assumed a fixed $\bK$ and $p(\by) = \normal_n(\by|\bfzero,\bK)$, but in practice $\bK = \bK_{\bftheta}$ and hence $p_{\bftheta}$ often depend on unknown parameters $\bftheta$. Our CVecchia approximation $\hat p(\by) = \prod_{i=1}^n p(y_i | \by_{c(i)})$ depends on $\bftheta$ both via $p_{\bftheta}$ and via the correlation distance $\tau_C^{\bftheta}$ in \eqref{eq:corrdist} used to determine the MM ordering of $y_1,\ldots,y_n$ and the NNs in the $c(i)$. To emphasize this, we will sometimes use $\hat p_{\bftheta_1}^{\bftheta_2}(\by)$ to denote a CVecchia approximation of $p_{\bftheta_1}$ based on $\tau_C^{\bftheta_2}$.

For frequentist inference, \citet{Guinness2019} proposed to find the maximum likelihood estimator of $\bftheta$ by optimizing the Vecchia loglikelihood via Fisher scoring. Given that
\begin{equation}
\label{eq:logvecchia}
\textstyle
\log \hat p(\by) = \sum_{i=1}^n \log p(y_i | \by_{c(i)}) = \sum_{i=1}^n \left( \log p(\by_{\tilde{c}(i)}) - \log p(\by_{c(i)}) \right),
\end{equation}
the score $\bg^{(k)}$ and the Fisher information $\bM^{(k)}$ of $\hat p(\by)$ at the $k$th iteration of the Fisher-scoring algorithm can be computed by addition and subtraction of the score and Fisher information of each of the $2n$ normal distributions of dimension at most $m+1$ on the right-hand side of \eqref{eq:logvecchia}. The parameter vector is then updated as $\bftheta^{(k+1)} = \bftheta^{(k)} + (\bM^{(k)})^{-1} \bg^{(k)}$.

For CVecchia, we propose to use a modified Fisher-scoring algorithm, where we now compute $\bg^{(k)}= \frac{\partial}{\partial \bftheta}\log \hat p^{\tilde\bftheta^{(k)}}_{\bftheta}(\by) \vert_{\bftheta = \bftheta^{(k)}}$ and $\bM^{(k)} = - \mathbb{E}  \frac{\partial^2}{\partial \bftheta^2}\log \hat p^{\tilde\bftheta^{(k)}}_{\bftheta}(\by) \vert_{\bftheta = \bftheta^{(k)}}$ with fixed ordering and conditioning based on $\tau_C^{\tilde\bftheta^{(k)}}$. In other words, when computing derivatives of the CVecchia loglikelihood for the Fisher-scoring updates, the dependence of the ordering and conditioning on $\bftheta$ is ignored. Instead, the ordering and conditioning are updated based on $\tilde\bftheta^{(k)} = \bftheta^{(k)}$ after certain iterations $k \in \mathcal{G}$, and $\tilde\bftheta^{(k)}=\tilde\bftheta^{(k-1)}$ otherwise. For simplicity, we can update the ordering and conditioning at the end of each iteration, $\mathcal{G} = \{1,2,3,4,\ldots\}$. Alternatively, the computational cost can be reduced by setting $\mathcal{G} = \{1,2,4,8\ldots\}$ and thus skipping this update for exponentially increasing numbers of iterations, exploiting that the parameter values tend to change less and less with increasing iteration numbers. In either case, repeatedly updating the ordering and conditioning over the course of the Fisher-scoring algorithm did not introduce convergence problems in our numerical experiments.

As the Vecchia approximation implies a valid density $\hat p(\by) = \normal_n(\by|\bfzero,\hat\bK)$, it is also possible to carry out Bayesian inference on $\bftheta$, assuming a prior $p(\bftheta)$ has been specified. However, the dependence of C-MM and C-NN on $\bftheta$ again presents a challenge. In the context of a spatio-temporal covariance, \citet{Datta2016a} essentially proposed to approximate the posterior as $\hat p(\bftheta|\by) \propto p(\bftheta) \hat p^{\bftheta}_{\bftheta}(\by)$ based on C-NN, meaning that the conditioning sets $c(i)$ are recomputed for every $\bftheta$ at which the posterior is evaluated. However, in our exploratory studies, we found this approach to lead to unstable and sinuous approximate posteriors. Instead, we propose to first obtain a maximum likelihood or maximum a posteriori estimate $\hat\bftheta$ using Fisher scoring, as above. Then, we approximate the posterior as $\hat p(\bftheta|\by) \propto p(\bftheta) \hat p^{\hat\bftheta}_{\bftheta}(\by)$, with fixed correlation distance $\tau^{\hat\bftheta}_C$ and hence fixed C-MM and C-NN based on $\bftheta = \hat\bftheta$. This approach leads to smooth posteriors, as illustrated in Section \ref{sec:sim_noise}.

\subsection{Prediction \label{sec:prdtn}}

Our method can be used for accurate and efficient prediction of an unobserved vector $\by^{*} = (y^{*}_{1} , \ldots , y^{*}_{n^{*}})$ with $\big( \by^{\top} , \by^{* \top} \big)^{\top} \sim \normal_{n + n^{*}} \left(\bfzero, \bK^{all} \right)$. For prediction and uncertainty analysis, the goal is to obtain the joint posterior predictive distribution $p(\by^{*} | \by)$.
Following \citet{Katzfuss2018}, we apply a Vecchia approximation to $\big( \by^{\top} , \by^{*}{}^\top \big)^{\top}$ with the entries of $\by^{*}$ ordered after those of $\by$ to obtain a CVecchia approximation of the posterior predictive distribution,
\begin{equation}
\label{eq:predvecchia}
\hat p(\by^{*} | \by) = \prod_{i=1}^{n^{*}} p(y^{*}_i | \by_{c_o(i)} , \by^{*}_{c_u(i)}) = \normal_{n^{*}} (\bfmu^* , \hat{\bK}^*),
\end{equation}
where $y^{*}_{1} , \ldots , y^{*}_{n^{*}}$ are assumed to follow a restricted C-MM ordering, which is obtained from a C-MM ordering of all (observed and unobserved) variables under the restriction of having the observed variables be ordered first (in which case the ordering of the unobserved variables takes the distances to observed variables into account).
As recommended in \citet{Katzfuss2018}, we allow the unobserved variables to condition on both observed and (previously ordered) unobserved variables. Specifically, 
$y^{*}_i$ conditions on the nearest (in terms of correlation-based distance with respect to $\bK^{all}$) $m$ variables among $y_1,\ldots,y_n,y_1^*,\ldots,y_{i-1}^*$. For notational convenience, in \eqref{eq:predvecchia} the resulting conditioning set is split into indices $c_o(i)$ corresponding to observed variables and $c_u(i)$ corresponding to unobserved variables; either $c_o(i)$ or $c_u(i)$ can be an empty set for any $i$. 
Each of the conditional distributions in the product in \eqref{eq:predvecchia} can be computed in $\order\left( m^3 \right)$ time, resulting in fast prediction or joint simulation even for large $n$ and $n^{*}$.

In practice, $\bK^{all}$ will typically depend on unknown parameters $\bftheta$. Predictions can then be based on a frequentist estimate of $\bftheta$ or based on samples from the Bayesian posterior of $\bftheta$, which can be obtained using the observed data $\by$ as described in Section \ref{sec:estim}. Then, given a frequentist estimate $\hat\bftheta$, the posterior predictive distribution is obtained as $\hat p^{\hat\bftheta}_{\hat\bftheta}(\by^{*} | \by)$ using similar notation as in Section \ref{sec:estim}. Given samples $\bftheta^{(1)},\ldots,\bftheta^{(L)}$ from the posterior, we can account for posterior uncertainty in $\bftheta$ and obtain an averaged posterior predictive distribution 
$\hat p(\by^{*} | \by)= (1/L) \sum_{l=1}^L \hat p_{\bftheta^{(l)}}^{\hat\bftheta}(\by^{*} | \by)$, where $\hat\bftheta$ is again a
maximum likelihood or maximum a posteriori estimate.

\subsection{Noise \label{sec:noise}}

The methods discussed so far are most appropriate if $\by$ is observed without noise. However, data in many application areas are typically modeled as a GP with additive noise. Suppose now that we observe $\bz = \by + \bfepsilon$ with $\bfepsilon = (\epsilon_1 , \ldots , \epsilon_n)^{\top} \sim \normal_n(\bfzero,\bD)$, where $\bD$ is diagonal.

A straightforward way of extending our methods to this noisy setting is to apply the same CVecchia approach to the covariance matrix of $\bz$, which is $\bfSigma = \bK + \bD$. However, in this approach the noise terms weaken the screening effect and hence an accurate CVecchia approximation will often require a larger $m$ than in the noise-free case. Interestingly, if the signal and noise variances are both constant (i.e., $\bK_{ii} = \bK_{jj}$ and $\bD_{ii} = \bD_{jj}$ for all $i,j$), then C-MM and C-NN do not depend on the noise variance (even if it is zero). This can be seen by noting that $\tau_C (i, j) \le \tau_C (i, k)$ if and only if $\tau^{+D}_C (i, j) \le \tau^{+D}_C (i, k)$, where $\tau^{+D}_C(i,j) = (1-|\rho^{+D}_{ij}|)^{1/2}$ with $\rho^{+D}_{ij} = \bfSigma_{ij}/(\bfSigma_{ii}\bfSigma_{jj})^{1/2}$ for $i,j \in \mathcal{I}$.
For varying noise variances, high-noise observations move farther away from other observations in terms of correlation distance, and so they are less likely to be included in conditioning sets; this makes intuitive sense, in that their high noise means that they contain less information about $\by$.

An alternative way of extending our methods to the noisy setting is to apply CVecchia to the (now latent) noise-free variables $\by$ as before and then add noise. In other words, we set $\hat\bfSigma = \hat\bK + \bD$, where $\hat\bK$ is obtained using CVecchia as in previous sections. While this is conceptually simple, inference then requires obtaining the Cholesky factor of the posterior precision matrix $var(\by|\bz)^{-1} = \hat\bK^{-1} + \bD^{-1}$, which can be very expensive due to fill-in. Fortunately, the computational speed of CVecchia can be maintained without introducing meaningful additional approximation error by approximating the Cholesky factor using an incomplete Cholesky factorization (IC), as proposed for EVecchia in \citet{Schafer2020}. This approach is useful both for parameter inference based on evaluating the CVecchia likelihood and for making predictions.
We demonstrate numerically in Section \ref{sec:sim_noise} that this IC-based approach can by highly accurate in the context of CVecchia as well.


\section{Examples and numerical comparisons \label{sec:simul}}

We conducted simulation experiments to demonstrate that CVecchia is widely applicable and highly accurate. Specifically, we considered anisotropic, nonstationary, multivariate, and spatio-temporal GPs, and an example without any explicit inputs.
We begin by assuming that the covariance matrices are known; then, we demonstrate parameter estimation and prediction using our methods. Throughout, our proposed CVecchia approach is denoted by C-MM + C-NN. We compared to existing or other reasonable competing Vecchia approximations, which necessarily differ between simulation scenarios, because none of them are meaningfully applicable across all the scenarios.
We compared the different Vecchia methods in terms of the KL-divergence between the exact distribution $\mathcal{N}(\mathbf{0}, \bK)$ and the approximate distribution $\mathcal{N}(\mathbf{0},\hat\bK)$, averaged over 10 simulations in settings with known covariance structure and over 200 simulations in parameter-inference or prediction settings. Comparisons are carried out as a function of $m$, as all considered Vecchia methods become more accurate and more computationally expensive as $m$ increases, with a time complexity of $\order(nm^3)$.

\subsection{Anisotropic and nonstationary GPs \label{sec:sim_nonstat}}

We considered nonstationary GPs at $n = 30^2 = 900$ inputs selected uniformly at random on the unit square, $\inputspace = [0,1]^2$.
We compared various combinations of ordering (E-MM, C-MM, X-ord, Y-ord) and conditioning (E-NN, C-NN) schemes, where X-ord and Y-ord denote ordering by the first or second coordinate of the input space, respectively.
EVecchia corresponds to E-MM + E-NN. \citet{Vecchia1988}'s original approach is given by Y-ord + E-NN.

We used a nonstationary Mat\'ern covariance function \citep{Stein2005,Paciorek2006}:
\begin{equation}
\textstyle K(\bx , \bx') = \sigma^2 \frac{\left|\bA(\bx)\right|^{1/4}\left|\bA(\bx')\right|^{1/4}}{|\tilde\bA(\bx,\bx')|^{1/2}} \ M_{\frac{\nu(\bx)+\nu(\bx')}{2}} \big(\big((\bx-\bx')^{\top}\tilde\bA(\bx,\bx')^{-1} (\bx-\bx')\big)^{1/2}\big), \quad \bx , \bx' \in \inputspace,
\label{eq:nonstat}
\end{equation}
where $M_{\nu}(0)=1$, $M_{\nu}(x) =x^\nu B_\nu (x)$ for $x > 0$, $B_{\nu}$ is a modified Bessel function of order $\nu$, $\nu: \inputspace \rightarrow \mathbb{R}^+$ is the smoothness, $\bA : \inputspace \rightarrow \mathbb{R}^{d \times d}$ is a (positive definite) anisotropy matrix, and  $\tilde\bA(\bx,\bx') = (\bA(\bx) + \bA(\bx'))/2$. For simplicity, we assumed $\sigma=1$.

We considered the following settings as special cases of \eqref{eq:nonstat}:
\begin{description}
\item[Anisotropic:] $\nu(\bx) \equiv 0.5$, $\bA(\bx) \equiv 10^{-2} \,  \diag(a^{-2},1)$, where $a$ is the degree of anisotropy. 
\item[Varying smoothness:] $\nu(\bx) = 0.2 + 1.3 x_1$ (i.e., varying as a function of the first coordinate), $\bA \equiv 10^{-2} \,  \diag(1,1)$.
\item[Varying rotation:] $\nu(\bx) \equiv 0.5$, $$\bA(\bx) = \begin{pmatrix} \cos{\eta(\bx)} & \sin{\eta(\bx)}\\ -\sin{\eta(\bx)} & \cos{\eta(\bx)}\end{pmatrix}^{\top} diag(10^{-4}, 10^{-2}) \begin{pmatrix} \cos{\eta(\bx)} & \sin{\eta(\bx)}\\ -\sin{\eta(\bx)} & \cos{\eta(\bx)}\end{pmatrix}$$ is a rotation matrix with spatially varying angle $\eta(\bx) = \frac{\pi x_1}{2}$. 
\end{description}

\begin{figure}
    \centering
    \includegraphics[width=\textwidth]{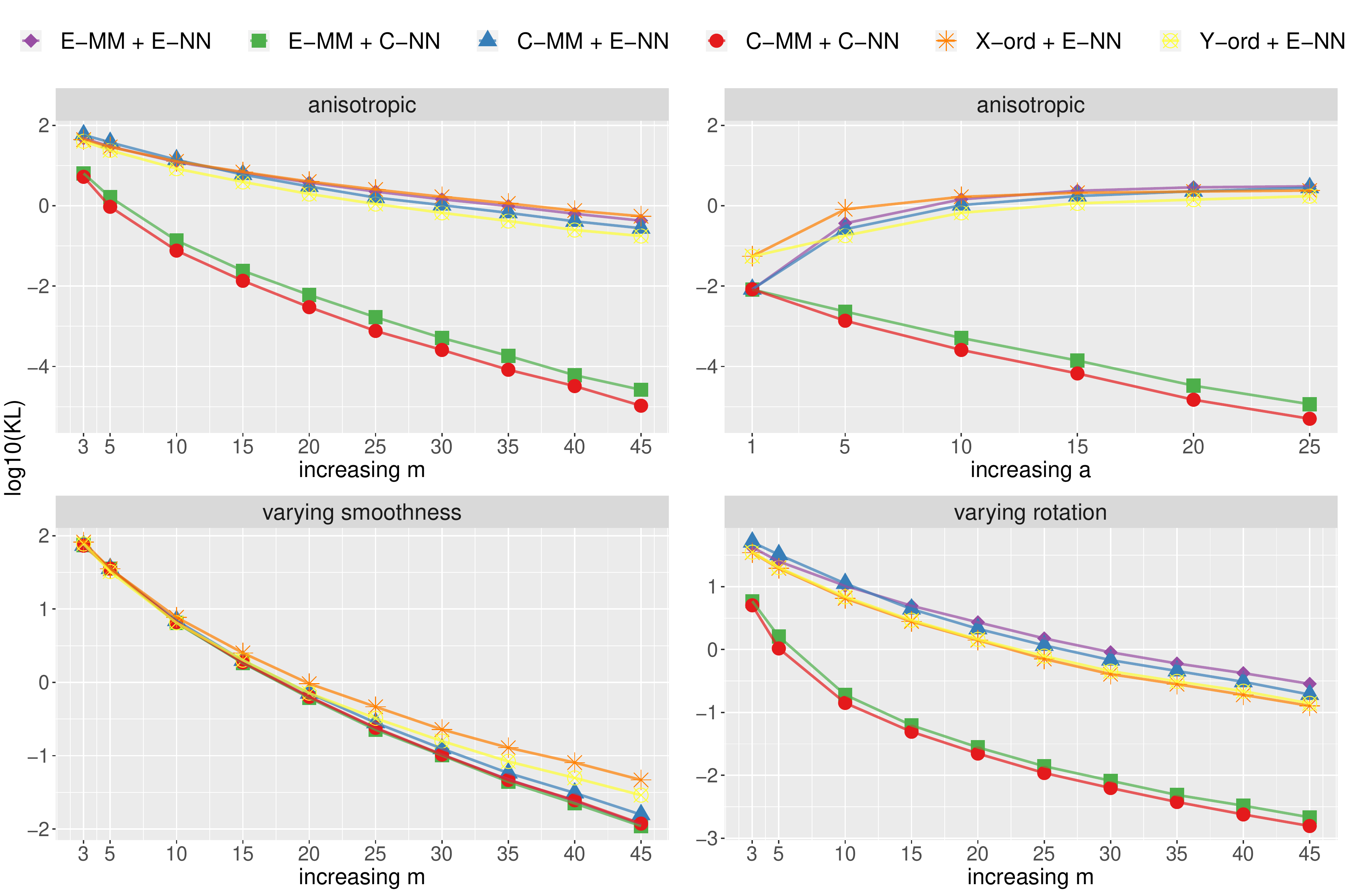}
    \caption{Log10-scale KL divergences between true and approximate likelihoods of GPs with the degree of anisotropy $a = 10$ for increasing size of conditioning sets $m$ \textbf{(top left)}, with $m = 30$ for increasing $a$ \textbf{(top right)}, with varying smoothness parameter $\nu = 0.2 + 1.3x_1$ (a function of the first coordinate) for increasing $m$ \textbf{(bottom left)}, and with varying rotation angle $\eta = \frac{\pi x_1}{2}$ for increasing $m$ \textbf{(bottom right)}}
    \label{fig:simulation_anisononst}
\end{figure}

In the anisotropic setting, the correlation-based distance $\tau_C(i,j)$ is a strictly increasing function of $\|\tilde\bx_i - \tilde\bx_j\|$, where $\tilde\bx_i = \bA^{-1/2}\bx_i$, because $B_\nu(\cdot)$ is strictly decreasing; thus, CVecchia is equivalent to EVecchia applied to the transformed inputs $\tilde\bx_1,\ldots,\tilde\bx_n$.
For varying smoothness and rotation, the transformed space is not easily identified.
However, as shown in Figure \ref{fig:simulation_anisononst}, C-NN was always more accurate than E-NN. In addition, using C-MM instead of E-MM led to further improvements for the anisotropic and varying-rotation setting. The improvement of CVecchia over existing methods was especially pronounced for strong anisotropy (i.e., large $a$) and for varying rotation.

\subsection{Multivariate GP \label{sec:sim_multiv}}

We considered a $p$-variate GP, $\by(\cdot) = (y^{(1)}(\cdot),\ldots,y^{(p)}(\cdot))^\top \sim \GP(0,K)$, with a cross-covariance function based on a latent dimension separating the processes \citep{Apanasovich2010},
\[
K_{i,j}(\bx,\bx') = \cov\big(y^{(i)}(\bx),y^{(j)}(\bx') \big) = \sigma^2 \exp\big( - \|\tilde{\bx}_i - \tilde{\bx}_j'\|/r\big), \qquad \bx,\bx' \in \inputspace, \hspace{0.35em} i,j \in \{1,\ldots,p\},
\]
where $\tilde{\bx}_i = \left(\bx^{\top} , \nu_i \right)^{\top} \in \mathbb{R}^{2+1}$, and
$\nu_i$ represents the location of the $i$-th component of the multivariate GP in the latent dimension. Thus, the dependence between $y^{(i)}(\cdot)$ and $y^{(j)}(\cdot)$ decreases with their latent distance $|\nu_i - \nu_j|$.
We assumed $\sigma^2 = 1$, $r = 0.1$, and $\nu_1 = 0$.
We considered a total of $n$ observations stacked into a vector $\by = (\by^{(1)}{}^\top,\ldots,\by^{(p)}{}^\top)^\top$, where $\by^{(j)} = (y^{(j)}_1,\ldots,y^{(j)}_{n_j})^\top$ with $y^{(j)}_i = y^{(j)}(\bx^{(j)}_i)$, and $n= \sum n_j$.

Here, $\tau_C(i,j)$ is a strictly increasing function of $\|\tilde\bx_i - \tilde\bx_j\|$, and so CVecchia is equivalent to EVecchia applied to the transformed inputs $\tilde\bx_1,\ldots,\tilde\bx_n$ in the expanded $(2+1)$-dimensional input space.
The competing methods considered in Section \ref{sec:sim_nonstat} are not directly applicable in this multivariate setting, and so we considered the following alternative approaches.
S-E-MM separately orders the entries of each $\by^{(j)}$ according to an MM ordering of the corresponding inputs $\bx^{(j)}_1,\ldots,\bx^{(j)}_{n_j}$, and then orders $\by^{(1)}$, then $\by^{(2)}$, and so forth, in $\by$.
To construct conditioning sets of size $m$, J-E-NN considers the nearest $m$ inputs in $\inputspace$ among all previously ordered variables in the joint vector $\by$, while S-E-NN carries out nearest-neighbor conditioning separately for each $\by^{(1)},\ldots,\by^{(p)}$.
D-E-NN divides $m$ by $p$ and finds the $m/p$ nearest previously ordered neighbors among each of the components $\by^{(1)},\ldots,\by^{(p)}$ (according to their inputs in $\inputspace$).

We compared these various Vecchia approaches for bivariate ($p=2$) and trivariate ($p=3$) GPs, with each process observed at $n_j = 400$ randomly sampled locations in $\inputspace$.
In both cases, we assumed that the processes were observed in a misaligned manner (i.e., $\bx^{(j)}_i \neq \bx^{(k)}_{i}$ for $j \neq k$).
As shown in Figure~\ref{fig:simulation_multivariate}, C-NN outperformed other conditioning approaches; C-MM provided additional improvements in some settings over S-E-MM.
We also considered the setting of identical observation locations for the different processes (i.e., $\bx^{(j)}_i = \bx^{(k)}_{i}$), but the results were very similar to the misaligned case and are hence not shown.

\begin{figure}
    \centering
    \includegraphics[width=\textwidth]{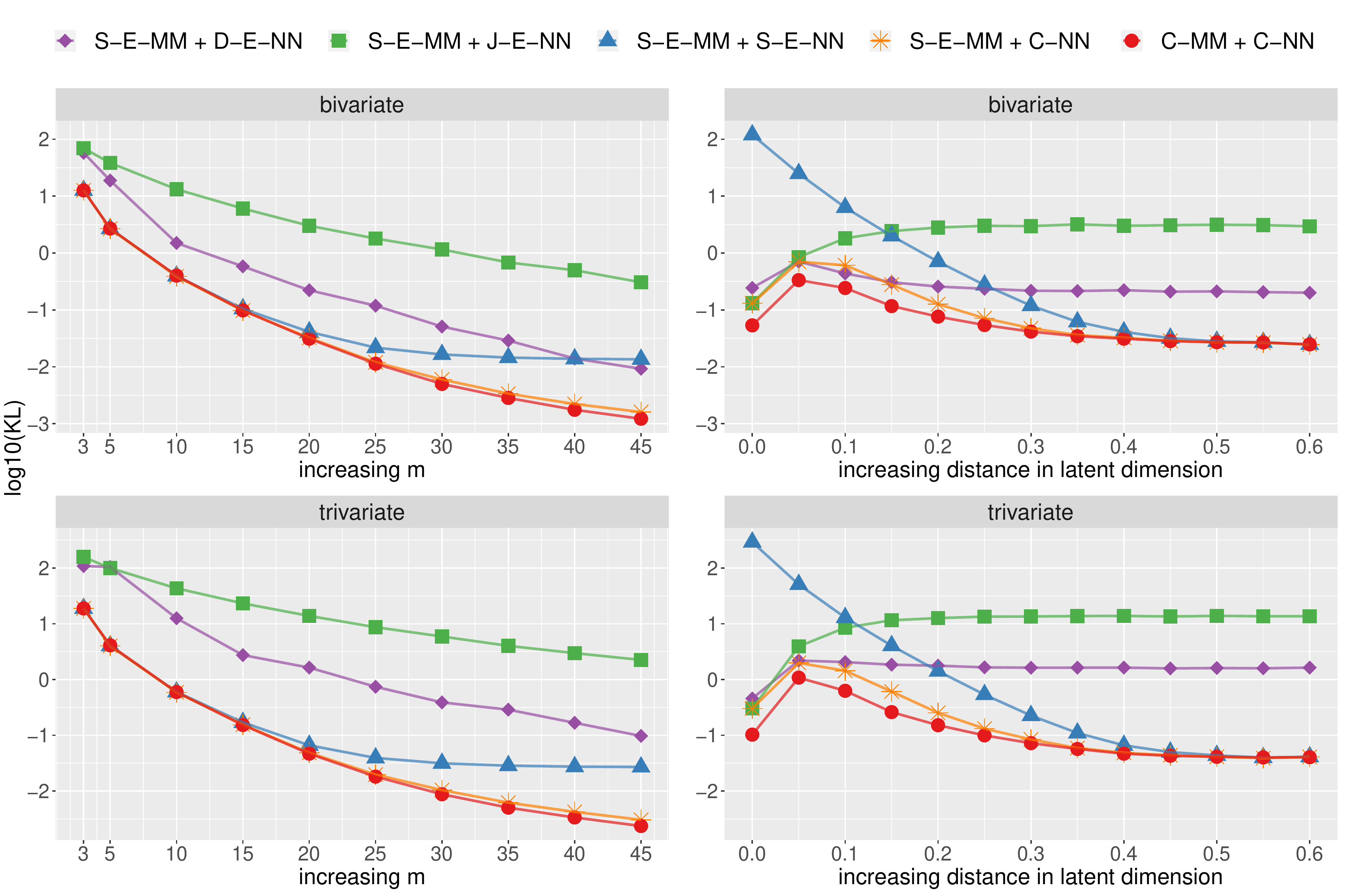}
    \caption{Log10-scale KL divergences between true and approximate likelihoods of distinctly observed bivariate GPs with distance in latent dimension $\Delta = |\nu_1 - \nu_2| = 0.4$ for increasing size of conditioning sets $m$ \textbf{(top left)} and with $m=20$ for increasing $\Delta$ \textbf{(top right)}, and distinctly observed trivariate GPs with distance in latent dimension $\Delta = 0.4$, where $\nu_j = (j-1)\Delta$, for increasing $m$ \textbf{(bottom left)} and with $m=20$ for increasing $\Delta$ \textbf{(bottom right)}}
    \label{fig:simulation_multivariate}
\end{figure}

\subsection{Spatio-temporal GP \label{sec:sim_spacetime}}

We considered a spatio-temporal GP indexed by a space-time input coordinate $\bx = (\bs^\top,t)^\top$, where we assumed that space is scaled to the unit square, $\bs \in [0, 1]^2$, and time is scaled to the unit interval, $t \in [0, 1]$.
We considered a space-time covariance function of the form
\begin{equation}
    \label{eq:spacetime}
K \left( \bx , \bx'\right) = \sigma^2 \exp ( -\| (\bs - \bs') \|/r_s - |t-t'|/r_t) = \sigma^2 \exp ( -\| \bA^{-1}(\bx - \bx') \| ), 
\end{equation}
where $r_s$ and $r_t$ are the spatial and temporal range parameters, and $\bA = diag(r_s, r_s, r_t)$. We assumed that $\sigma^2 = 1$, $r_s = 0.1$, and $r_t = 1.0$.

Here, $\tau_C(i,j)$ is a strictly increasing function of $\|\tilde\bx_i - \tilde\bx_j\|$, where $\tilde\bx_i = \bA^{-1/2}\bx$, and so CVecchia is equivalent to EVecchia applied to the transformed inputs $\tilde\bx_1,\ldots,\tilde\bx_n$.
As space and time are not commensurable, the previous competing methods are again not meaningful. We considered ordering by time (T-ord), and conditioning on the NN in time (T-NN). Note that, when inputs are taken at the same time point, T-ord orders the inputs according to the values of the second spatial coordinate. If these values are again the same, it uses the values of the first coordinate. Further, we considered E-NN based on the distance of the (unit-scaled) space-time coordinates, $\|\bx - \bx' \|$. To our understanding, the correlation-based conditioning approach proposed in \citet{Datta2016a} corresponds to T-ord + C-NN.

As illustrated in Figure~\ref{fig:simulation_scenarios}, we simulated $n=900$ space-time observations on the unit cube according to four different simulation scenarios, the latter three of which were chosen to mimic common observation patterns for environmental data: 
\begin{description}
\item[Random] Space-time coordinates are selected uniformly at random, and so they are irregular in space and time.
\item[Station] Observations are obtained at 9 regular time points at 100 irregularly spaced ``monitoring stations.''
\item[Gridded] Observations are obtained at 9 regular time points on a regular grid of size $10 \times 10 = 100$ in space (e.g., mimicking output from climate models).
\item[Satellite] Similar to data from polar-orbiting satellites, at 900 regularly spaced time points, we have 90 observations along each of 5 one-dimensional tracks at two repeat cycles.
\end{description}

\begin{figure}
    \centering
    \includegraphics[width=\textwidth]{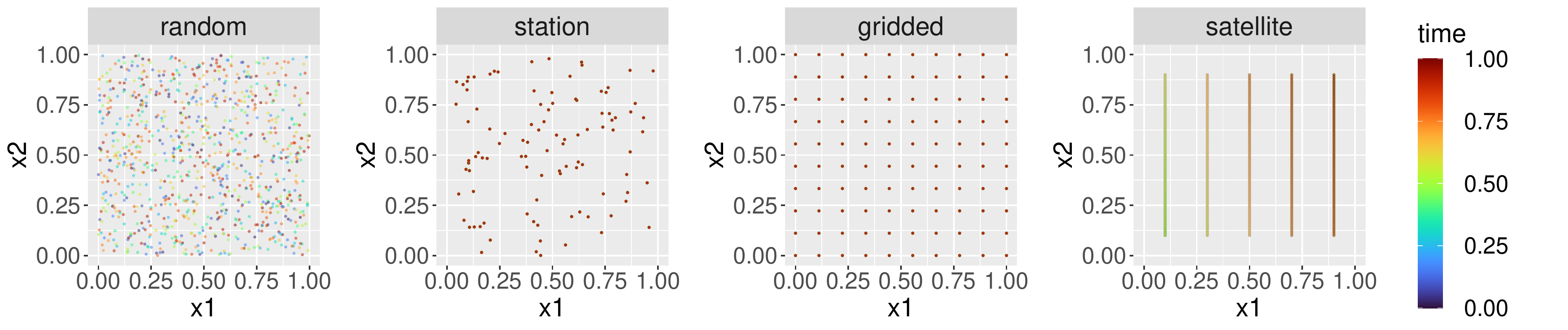}
    \caption{Inputs for spatio-temporal simulation scenarios plotted against spatial coordinates $x_1$ and $x_2$. The inputs are color-coded by time, although later time points exactly cover earlier time points in the Station and Gridded case, and also to some degree in the Satellite scenario due to its two repeat cycles.}
    \label{fig:simulation_scenarios}
\end{figure}

As shown in Figure \ref{fig:simulation_spacetime}, CVecchia outperformed the competing methods.

Note that we repeated the experiments from Sections~\ref{sec:sim_nonstat}--\ref{sec:sim_spacetime} for larger $n = 3,600$, but we found out that the shapes of the KL curves were very similar to those in Figures~\ref{fig:simulation_anisononst}, \ref{fig:simulation_multivariate} and \ref{fig:simulation_spacetime}.

\begin{figure}
    \centering
    \includegraphics[width=\textwidth]{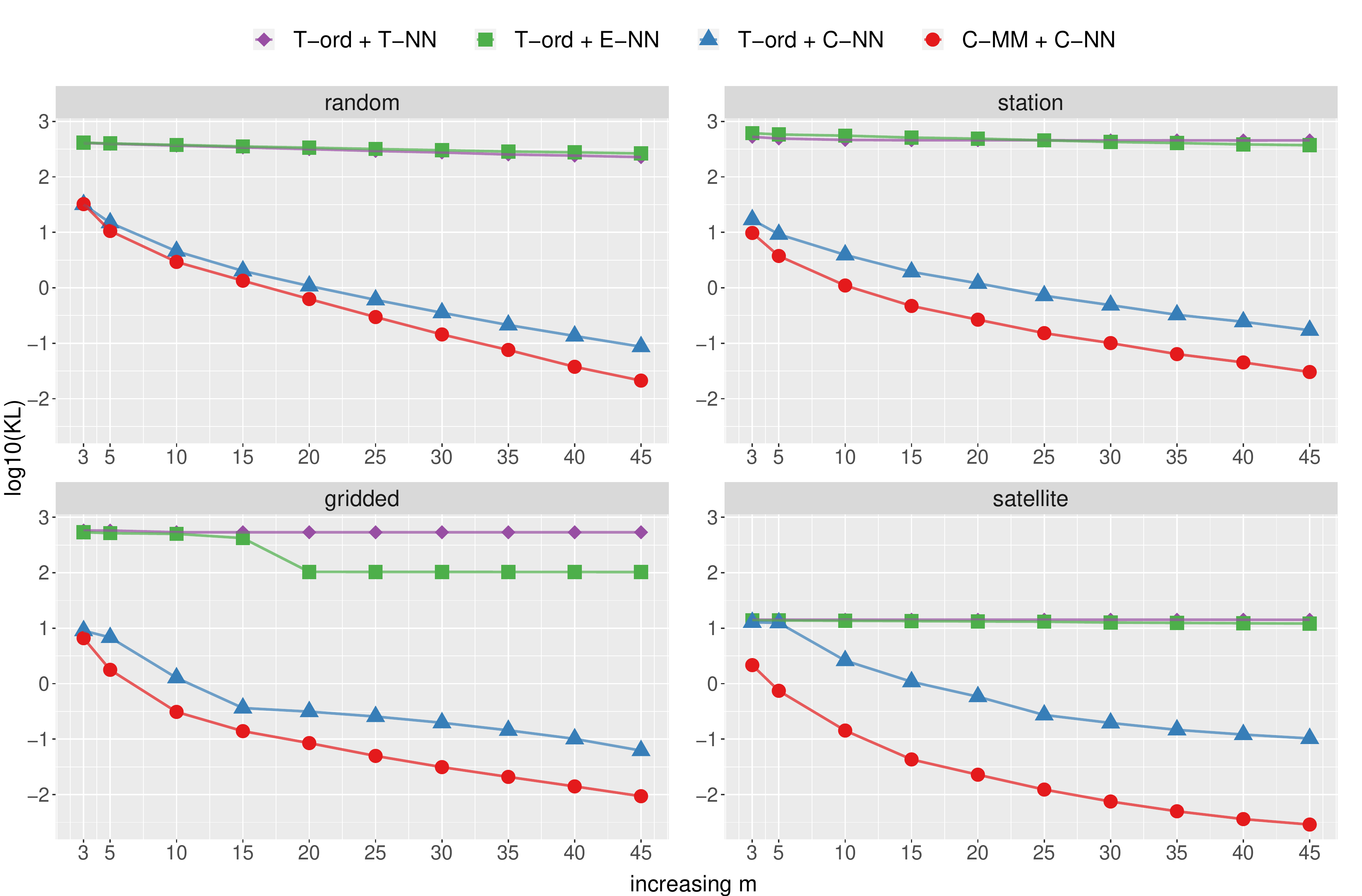}
    \caption{Log10-scale KL divergences between true and approximate likelihoods of spatio-temporal GPs (Figure \ref{fig:simulation_scenarios}) for increasing size of conditioning sets $m$}
    \label{fig:simulation_spacetime}
\end{figure}

\subsection{Gaussian hierarchical model \label{sec:sim_nonspatial}}

We have so far considered only cases in which covariance structure is computed based on inputs. In this subsection, we offer an example that has no inputs, so that CVecchia is applicable but EVecchia is not. Motivated by hierarchical models which are widely used for combining information and describing heterogeneity between sub-populations, we assumed that $\mu \sim \normal \left( 0 , \sigma_0^2 \right)$ and 
\[
\mu_{i_1 , \ldots , i_j} \ |\  \mu_{i_1 , \ldots , i_{j-1}} \overset{i.i.d.}{\sim} \normal( \mu_{i_1 , \ldots , i_{j-1}} , \sigma_j^2 )\ , \quad i_k = 1,2, \quad k=1,\ldots,j, \quad j = 1, \ldots , J,
\]
where $\sigma_0^2 = \sigma_1^2 = \ldots = \sigma_k^2 = 1$. We observe $\by = \{y_{i_1,\ldots,i_J}: i_k = 1,2, \, k=1,\ldots,J \}$ with $y_{i_1,\ldots,i_J} = \mu_{i_1,\ldots,i_J}$ at the finest level. This hierarchical model is illustrated for depth $J=3$ on the left side of Figure~\ref{fig:simulation_nonspatial}. 
We have $\cov \left( y_{i_1 , \ldots , i_J}, y_{l_1 , \ldots , l_J} \right) = \sum_{r=0}^{\alpha} \sigma_r^2$, where $\alpha$ is the level up to which $y_{i_1 , \ldots , i_J}$ and $y_{l_1 , \ldots , l_J}$ have a common ancestor.

We conducted a numerical comparison using depth $J=12$, and so $n=2^{12} = 4{,}096$.
Because the competing methods used in other experiments are again not directly applicable, we compared three variants of the Vecchia approximation in the right panel of Figure~\ref{fig:simulation_nonspatial}, where L-ord denotes lexicographic (or simply left-to-right) ordering, R-ord denotes random ordering, and R-N conditions on randomly selected previously ordered entries. We repeated R-ord + R-N 200 times, but interestingly the resulting KL divergences appear quite similar when plotted on the log scale. CVecchia strongly outperformed the other two methods.

\begin{figure}
    \centering
    \includegraphics[align=c,width=0.49\textwidth]{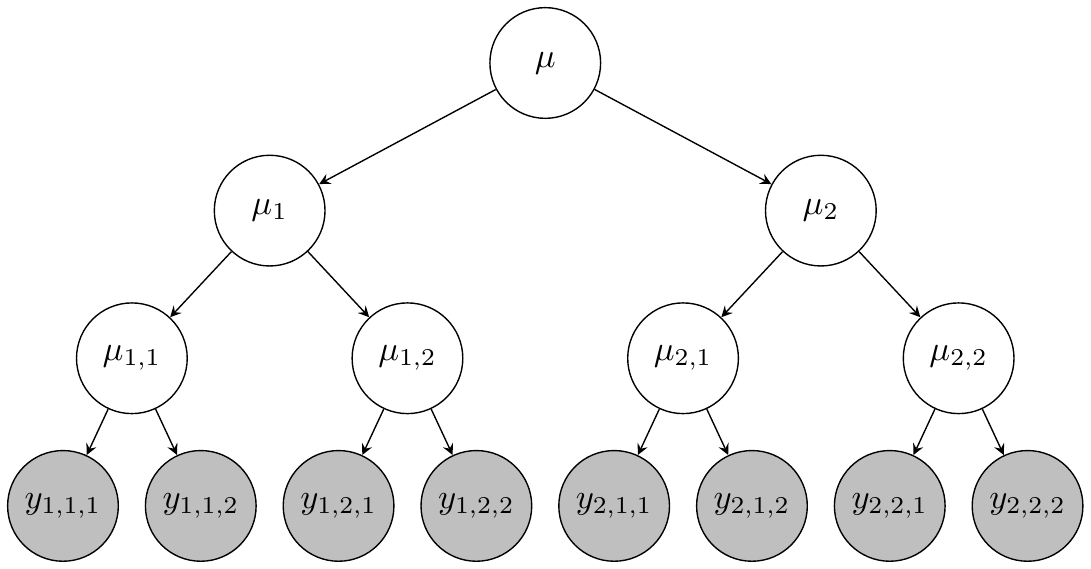}
    \includegraphics[align=c,width=0.5\textwidth]{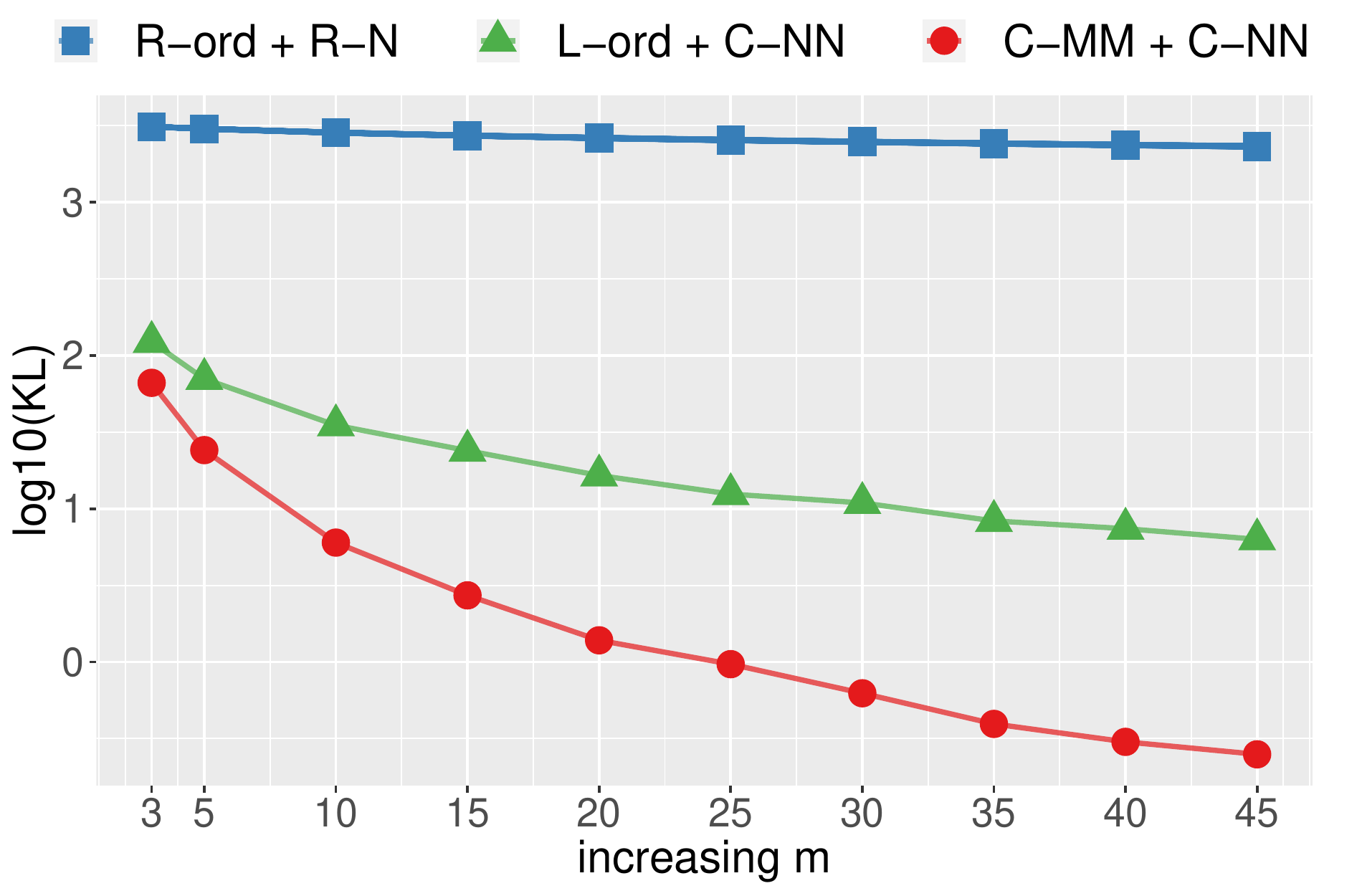}
    \caption{A graphical representation of the hierarchical normal model with $J=3$ \textbf{(left)} and log10-scale KL divergences between true and approximate likelihoods with $J=12$ for increasing $m$ \textbf{(right)}}
    \label{fig:simulation_nonspatial}
\end{figure}

\subsection{Parameter estimation \label{sec:sim_estim}}

We examined the performance of frequentist parameter estimation using Fisher scoring (Section \ref{sec:estim}) in the Station and Satellite space-time scenarios of Section \ref{sec:sim_spacetime}. The task was to estimate the range parameters $r_s$ and $r_t$. We updated the ordering and conditioning at every Fisher-scoring iteration ($\mathcal{G} = \{1,2,3,\ldots\}$).

We compared the different approximation methods described in Section \ref{sec:sim_spacetime}. For reference, we also considered ``optimal'' parameter estimation using the exact GP without Vecchia approximation (or, equivalently, a Vecchia approximation with $m=n-1$). The methods were compared in terms of the average KL divergence between the true distribution (using the true parameter values) and the approximate distribution (using each method's estimated parameters). We also computed the root mean squared difference (RMSD) between the values of $r_s$ and $r_t$ as estimated by the exact GP and as estimated by the different Vecchia approximations.

As shown in Figure \ref{fig:simulation_fisher}, CVecchia produced by far the most accurate estimated distributions, which were similar to those based on the exact GP for $m\geq 25$. While the RSMDs were quite noisy, despite averaging over 200 simulated datasets, CVecchia also generally performed best in terms of RMSD.

\begin{figure}
     \centering
     \begin{subfigure}[b]{\textwidth}
         \centering
         \includegraphics[width=\textwidth]{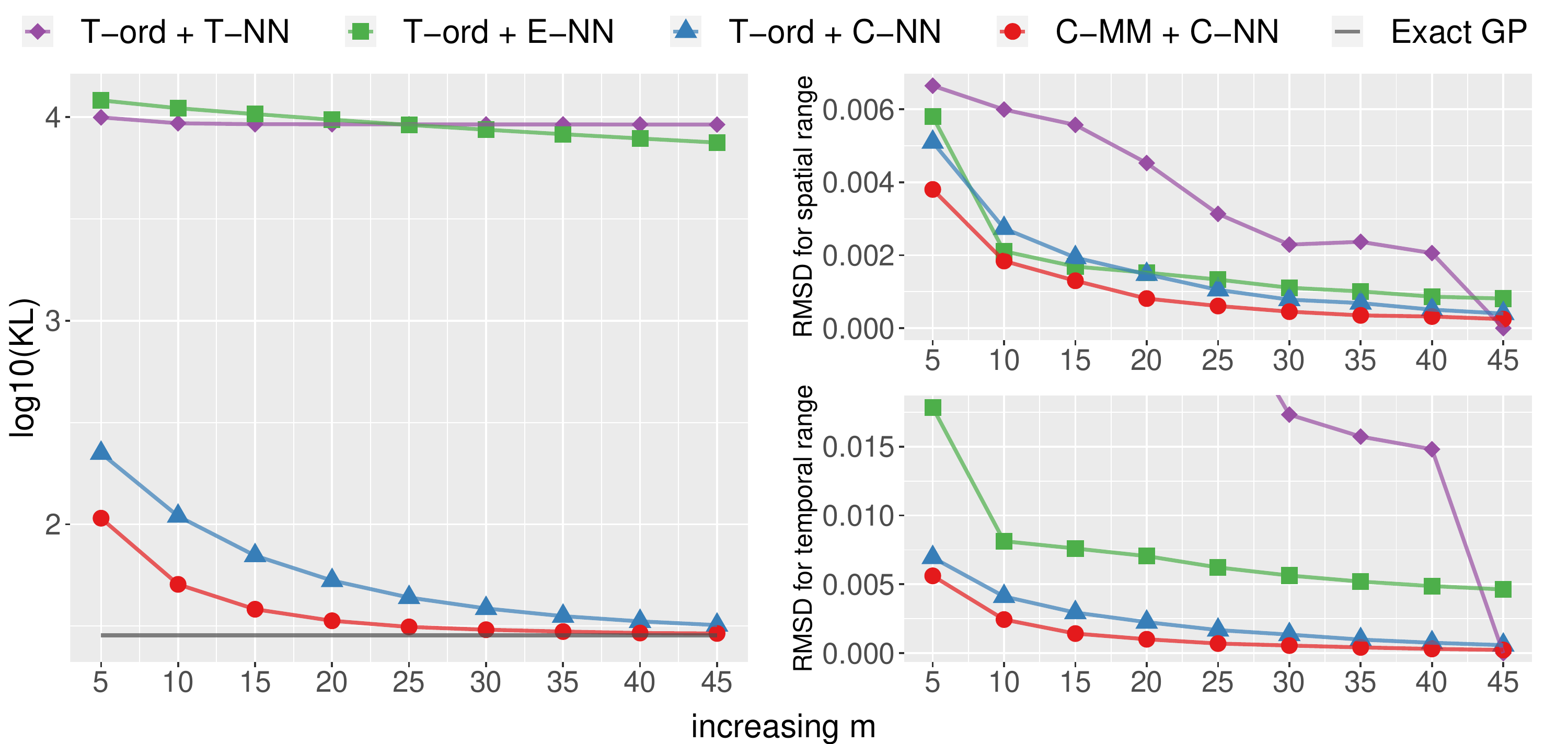}
         \caption{Station scenario}
         \label{fig:simulation_fisher_monitoringstation}
     \end{subfigure}

\bigskip\smallskip

\begin{subfigure}[b]{\textwidth}
         \centering
         \includegraphics[width=\textwidth]{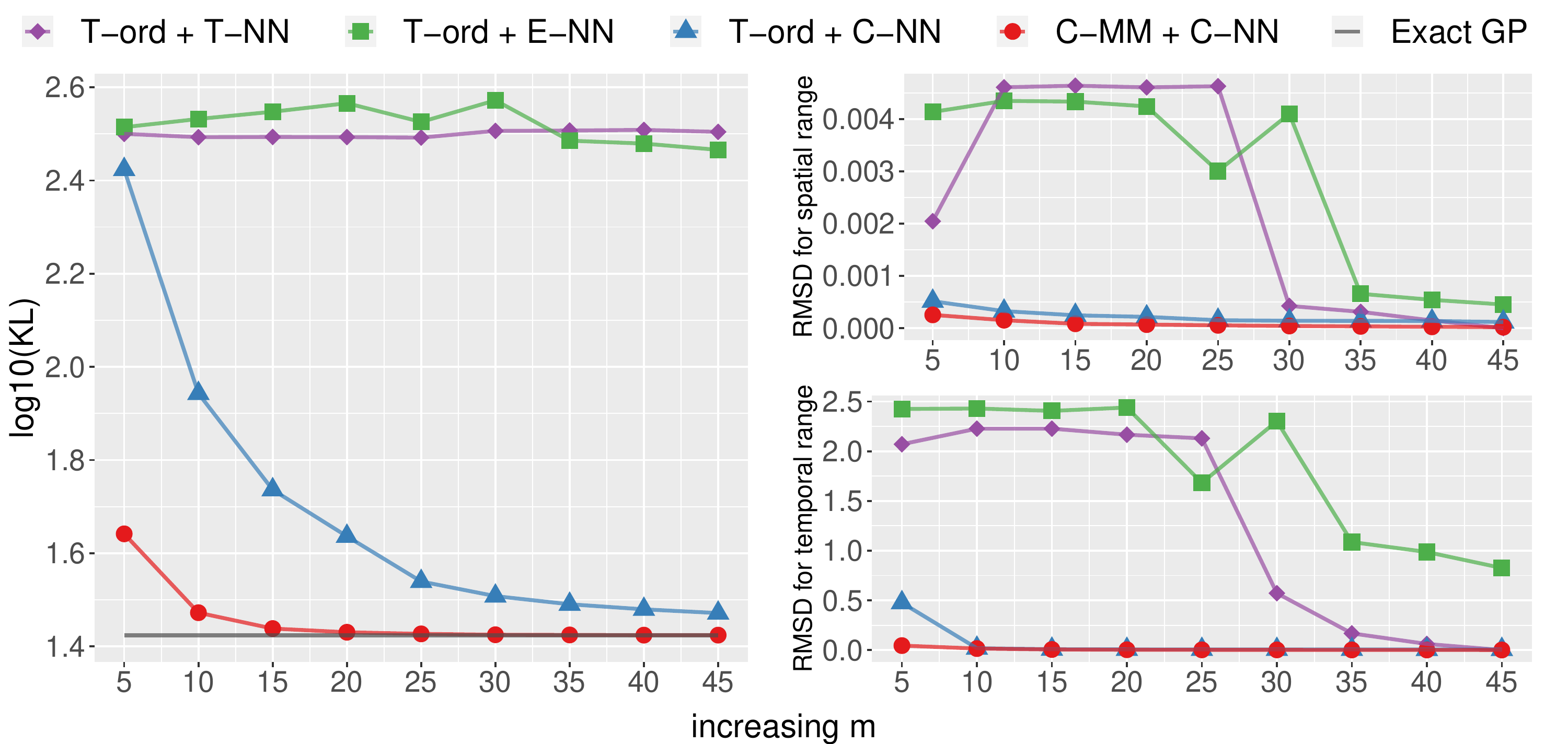}
         \caption{Satellite scenario}
         \label{fig:simulation_fisher_satellite}
     \end{subfigure}
     
     \medskip
     
        \caption{Performance in parameter estimation using Fisher scoring under two spatio-temporal scenarios (station and satellite) in Figure \ref{fig:simulation_scenarios}. Each subfigure contains three plots: Log10-scale KL divergences between true and estimated likelihoods \textbf{(left)} and root mean squared difference (RMSD) for spatial range parameter \textbf{(top right)} and for temporal range parameter \textbf{(bottom right)}, for increasing size of conditioning sets $m$. }
        \label{fig:simulation_fisher}
\end{figure}

\subsection{Prediction \label{sec:sim_predtn}}

To illustrate prediction performance, we again considered the Random, Station, and Satellite space-time scenarios from Section \ref{sec:sim_spacetime}. Of the 900 space-time observations, 100 were randomly selected as test data, and so the training data consisted of the remaining $n=800$ observations. To lessen the computational cost of our many comparisons, we assumed that the covariance parameters were known.

Figure~\ref{fig:simulation_prediction} shows the prediction performance for the 100 test data, as measured by the logarithmic score \citep[see][for details]{Gneiting2014} averaged over 200 simulation runs. In the Random and Station scenarios, CVecchia and T-ord + C-NN both performed well. In the Satellite scenario, CVecchia performed best.

\begin{figure}
    \centering
    \includegraphics[width=\textwidth]{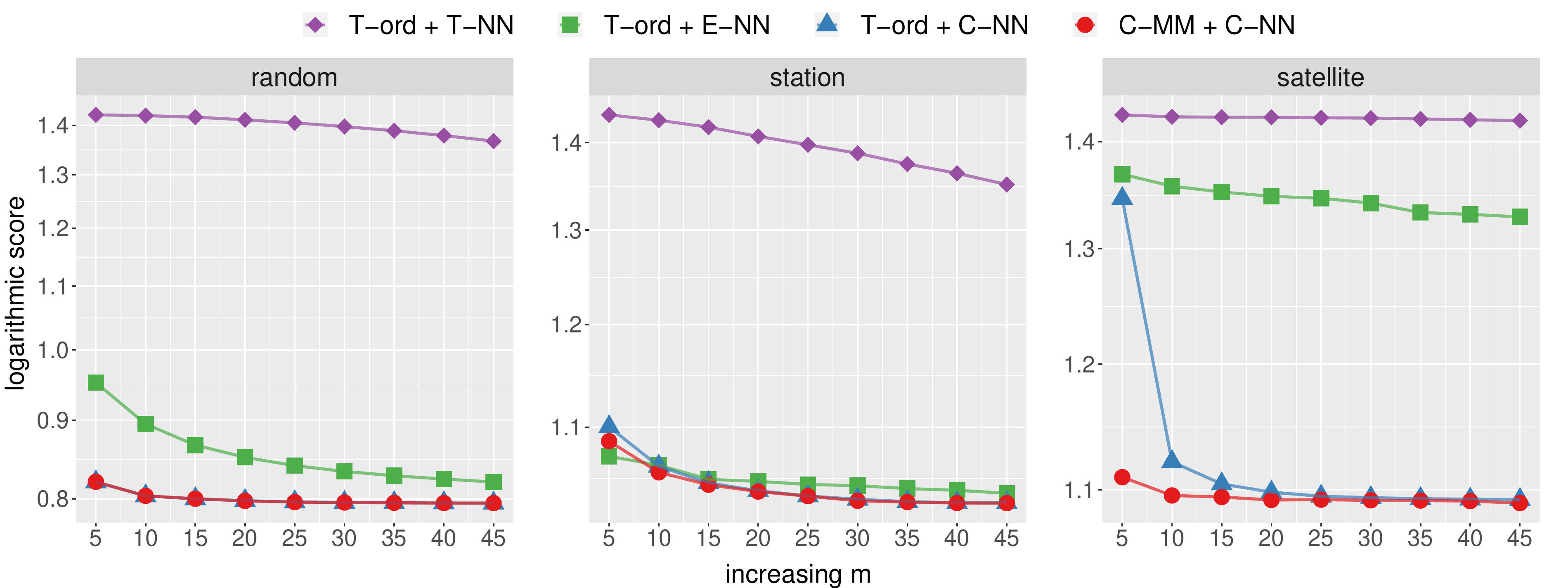}
    \caption{Logarithmic score for predictions under three space-time scenarios: Random \textbf{(left)}, Station \textbf{(center)}, and Satellite \textbf{(right)}. Note that the $y$-axes of the panels are on a log scale.}
    \label{fig:simulation_prediction}
\end{figure}

\subsection{Bayesian inference for noisy data \label{sec:sim_noise}}

We  considered Bayesian inference with CVecchia for noisy data under the Random, Station, and Satellite space-time scenarios of Section~\ref{sec:sim_spacetime}. The task was to calculate posterior densities of the range parameters $r_s$ and $r_t$. We assumed that the priors were
$
\log(r_s) \sim \normal \left( \log(0.1), 0.6^2 \right)
$
and 
$
\log(r_t) \sim \normal \left( \log(1.0), 0.6^2 \right)
$,
with constant noise variances, $\bD = (0.4) \bI_n$. Figure \ref{fig:simulation_noise} presents two different approaches described in Section \ref{sec:noise}: one is the naive approach that directly uses the covariance matrix of the noisy observations, and the other is the IC-based approach that applies CVecchia to the noise-free variables and then adds the noise. As claimed in Section~\ref{sec:noise}, Figure \ref{fig:simulation_noise} shows that, while CVecchia provided reliable approximate posteriors compared to the other methods, the IC-based approach provided further improvements.
C-MM and C-NN were fixed based on the true values of $\bftheta$; we also tried updating C-MM for each evaluated $\bftheta$ value, but this resulted in unstable posteriors.

\begin{figure}
    \centering
    \includegraphics[width=\textwidth]{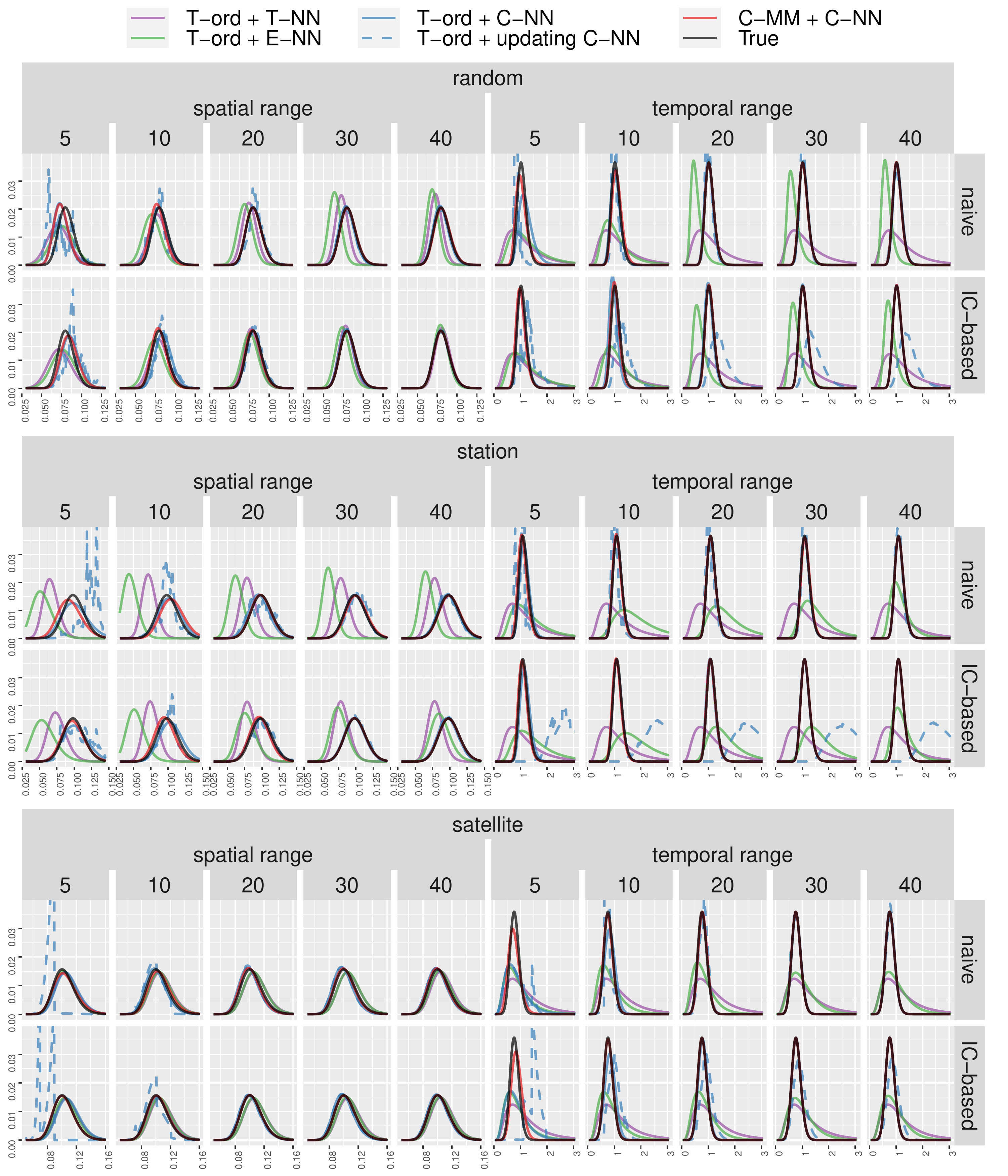}
    \caption{True and approximate posterior densities of spatial and temporal ranges under three space-time scenarios from Figure \ref{fig:simulation_scenarios}: Random \textbf{(top)}, Station \textbf{(center)}, and Satellite \textbf{(bottom)}. For each scenario, the left (right) five columns are posteriors of the spatial (temporal) range parameter. For each range parameter, the first (second) row presents posteriors with the naive (IC-based) approach for size of conditioning sets $m = 5, 10, 20, 30, 40$. For $m \ge 10$, some lines are not visible because they are covered by the (exact) black lines.}
    \label{fig:simulation_noise}
\end{figure}


\section{Application to real data \label{sec:realdat}}

We assessed the use and efficacy of CVecchia to fit and predict regional climate model (RCM) outputs. The North American Regional Climate Change Assessment Program \citep[NARCCAP; ][]{mearns2009regional} is a research program designed to (1) provide high-resolution projections of climate change, (2) investigate uncertainties in regional climate change simulations based on different atmosphere-ocean general circulation models (AOGCMs), and (3) evaluate RCM performance over North America \citep{mearns2012north}. While the program ran 50-km spatial resolution simulation based on multiple RCMs driven by multiple AOGCMs, we only considered the Canadian regional climate model (CRCM) using the NCEP-DOE Reanalysis II (NCEP) as boundary conditions. The details on RCMs and AOGCMs in the NARCCAP are available from \href{https://www.narccap.ucar.edu/}{https://www.narccap.ucar.edu/}.

In particular, we studied a bivariate spatio-temporal dataset given by a maximum and minimum daily surface air temperature (tasmax and tasmin) for June--August 2001 (92 days) in the South region \citep[Arkansas, Kansas, Louisiana, Mississippi, Oklahoma, and Texas; see][]{karl1984regional}. Figure~\ref{fig:realdata_tas} shows tasmax and tasmin fields in the South region on selected days. The cartographic boundary files of the south region are available from \href{https://www.census.gov/geographies/mapping-files/time-series/geo/carto-boundary-file.html}{https://www.census.gov/}. The total sample size is $n_{total} = 78{,}384 \times 2 = 156{,}768$. We split the dataset into training ($n_{train} = 114{,}298$) and test ($n_{test} = 42,470$) sets in the following manner: (1) randomly select 12 locations for each time slice; (2) assign observations (for both variables) corresponding to space-time locations on the $5^2 \times 3$ space-time cube centered at the selected locations to the test set; and (3) assign the remaining space-time locations to the training set.

We fit a joint model of tasmin and tasmax using the training set and then carried out predictions on the test set. Let $\by_{tasmin}$ and $\by_{tasmax}$ be training vectors of tasmin and tasmax, respectively. We modeled them as
\begin{equation}
\begin{bmatrix} 
\by_{tasmin} \\ \by_{tasmax} \end{bmatrix} 
\sim \normal_{n + n} \left( 
\begin{bmatrix} \textbf{1} & \textbf{0} \\ \textbf{1} & \textbf{1} \end{bmatrix}
\begin{bmatrix} \beta_0 \\ \beta_1 \end{bmatrix}
,
\bK
\right)
\end{equation}
using a Mat\'ern covariance function with a different range parameter for each dimension (latitude, longitude, time, and latent dimension); that is, 
\[
\textstyle
\bK_{i, j} = K (\tilde\bx_i, \tilde\bx_j) = \sigma^2 \, \frac{2^{1-\nu}}{\Gamma (\nu)} \, \| \bA^{-1} (\tilde{\bx}_i - \tilde{\bx}_j) \|^{\nu} \, B_{\nu}( \| \bA^{-1} (\tilde{\bx}_i - \tilde{\bx}_j) \| ),
\]
$\tilde{\bx} = \left( \bx^{\top} , \xi \right)^{\top}$, $\bx$ is a space-time coordinate, $\xi$ is an indicator variable that indicates whether $\tilde\bx$ corresponds to tasmin, $\Gamma$ is the gamma function, $B_\nu$ is the modified Bessel function of the second kind, and $\bA = diag (r_{lat}, r_{lon}, r_{t}, r_{l}).$ Assuming that $\nu = 0.75$ (based on preliminary analyses) and a nugget of zero, we estimated the unknown parameters $\beta_0$, $\beta_1$, $r_{lat}$, $r_{lon}$, $r_{t}$, $r_{l}$ using the Fisher scoring approach described in Section~\ref{sec:estim}; the result is given in Table~\ref{tbl:realdata_estimates}.

Figure~\ref{fig:realdata_perf} shows the prediction performance for the test set, as measured by the root mean square prediction error (RMSPE), compared to five other Vecchia variants. S-E-MM + S-E-NN and S-E-MM + J-E-NN are from Section~\ref{sec:sim_multiv} and based on Euclidean distance between unit-scaled space-time coordinates. T-ord + T-NN is from Section~\ref{sec:sim_spacetime}. Note that T-ord separately orders observations of each temperature field by time and then joins them. S-C-NN carries out C-NN conditioning separately for each temperature field, while J-C-NN searches C-NN in the joint vector. We applied a grouping algorithm \citep{Guinness2016a} for improving computational efficiency to all methods except T-ord + T-NN, because interestingly it resulted in a doubling of the computational cost for that method.

CVecchia (C-MM + C-NN) provided the lowest RMSPE for any $m$ considered. The improvement was substantial, with CVecchia's accuracy with $m=10$ surpassing that of S-E-MM + J-E-NN with $m=50$, whose computational cost is roughly two orders of magnitude higher due to the cubic scaling in $m$. Moreover, as shown in the right panel of Figure~\ref{fig:realdata_perf}, CVecchia offered a better trade-off between run time and prediction accuracy. The run-time analysis was performed on a 64-bit workstation with 16 GB RAM and an Intel Core i7-8700K CPU running at 3.70 GHz.
We also carried out a comparison in terms of the logarithmic score, but the resulting curves looked almost identical to the RMSPE curves in Figure~\ref{fig:realdata_perf}.

\begin{table}
\centering
\caption{For the NARCCAP data, parameter estimates for the six methods ($m = 50$) using a Mat\'ern covariance function with a different range parameter for each dimension, smoothness $\nu = 0.75$, and zero nugget. Temperatures are in Kelvin, the spatial region is scaled to fit into the unit square (without changing its shape), and the time period is scaled to the unit interval.}
\label{tbl:realdata_estimates}
\begin{tabular}{l|ccccccc}
 & $\hat{\beta}_0$ & $\hat{\beta}_1$ & $\hat{\sigma}^2$ & $\hat{r}_{lat}$ & $\hat{r}_{lon}$ & $\hat{r}_{t}$ & $\hat{r}_{l}$ \\ \hline
 S-E-MM + S-E-NN & 278.513 & 12.084 & 61.983 & 0.819 & 0.742 & 0.036 & 2.000 \\
 S-E-MM + J-E-NN & 278.200 & 14.033 & 62.906 & 0.828 & 0.750 & 0.036 & 2.681 \\
 T-ord + T-NN & 275.143 & 13.613 & 70.134 & 0.890 & 0.792 & 0.001 & 2.000 \\
 T-ord + S-C-NN & 268.826 & 11.451 & 51.585 & 0.719 & 0.656 & 0.040 & 2.000 \\
 T-ord + J-C-NN & 266.677 & 12.198 & 51.331 & 0.716 & 0.654 & 0.040 & 2.498 \\
 C-MM + C-NN & 276.754 & 13.054 & 38.859 & 0.593 & 0.543 & 0.026 & 1.668
\end{tabular}
\end{table}

\begin{figure}
    \centering
    \includegraphics[width=\textwidth]{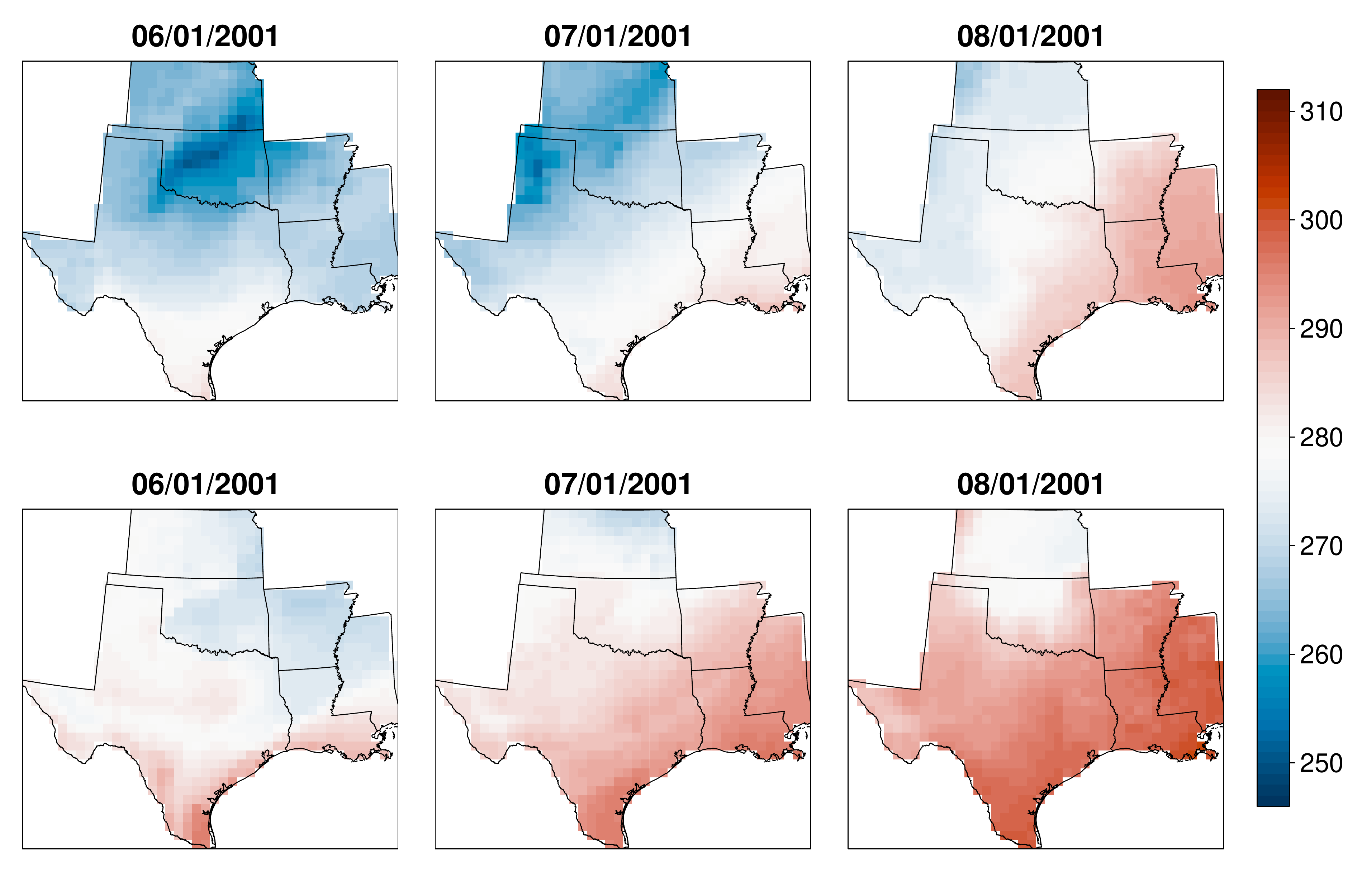}
    \caption{Minimum \textbf{(top row)} and maximum \textbf{(bottom row)} surface air temperature fields (in degrees Kelvin) in the South region (Arkansas, Kansas, Louisiana, Mississippi, Oklahoma and Texas) from NARCCAP}
    \label{fig:realdata_tas}
\end{figure}

\begin{figure}
    \centering
    \includegraphics[width=\textwidth]{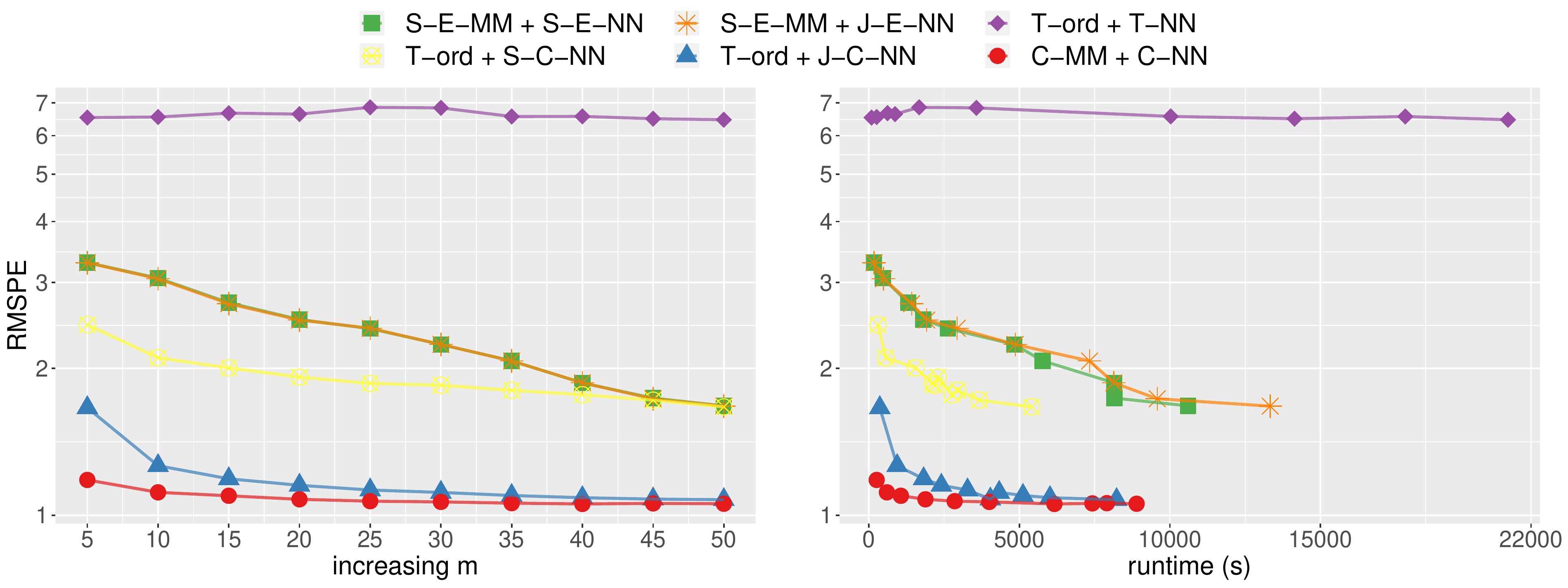}
    \caption{For the NARCCAP data, root mean squared prediction error (on a log scale) at held-out test points as a function of $m$ \textbf{(left)} and as a function of training time of the Fisher-scoring algorithm \textbf{(right)}}
    \label{fig:realdata_perf}
\end{figure}


\section{Conclusions \label{sec:conc}}

We have introduced CVecchia, a covariance approximation that results in a sparse inverse Cholesky factor, whose ordering and sparsity pattern are based on the correlation structure. For reducible GPs, CVecchia implicitly applies a Euclidean-based Vecchia approximation in a transformed input domain in which the GP is isotropic. CVecchia is applicable to any covariance matrix, and it even allows for likelihood-based inference on unknown covariance parameters. We numerically demonstrated the applicability of CVecchia to a variety of covariance structures, some of which had no applicable existing Vecchia approximations. In settings with suitable existing approximations, CVecchia strongly outperformed them.

Special cases of our general CVecchia idea have already been successfully employed in several applications (which were started later but completed earlier than the present paper): \citet{Katzfuss2020} used the idea to approximate anisotropic GPs for computer-model emulation in high input dimension; \citet{Messier2020} approximated spatio-temporal land-use regression for ground-level nitrogen dioxide; and in the context of nonparametric inference \citep{Kidd2020}, ideas related to CVecchia were used with sample correlations instead of parametric correlations.

While we have largely focused on geospatial settings here, CVecchia can also be applied to large-scale machine learning settings where input domains are not Euclidean and there is no explicit expression of covariance function. Examples include: multi-task learning \citep{williams2007multi,groot2011learning}, where multiple observations are collected from multiple related tasks and joint modeling utilizes intra- and inter-task relatedness; natural language processing \citep[NLP; see][for recent review]{min2021recent}, where words are represented in a latent vector space using word embedding methods \citep[e.g.,][]{beck2014joint,beck2017modelling,deriu2017leveraging} and CVecchia can be applied efficiently based on only geometric relations between word vectors; and modeling of multiple interacting latent chemical species in biochemical interaction networks \citep{gao2008gaussian}, where the covariance function of gene expression is an integral equation of inputs.


\footnotesize
\appendix
\section*{Acknowledgments}

Both authors were partially supported by National Science Foundation (NSF) Grant DMS--1953005. Katzfuss' research was also partially supported by NSF Grants DMS--1654083 and CCF--1934904. We would like to thank Florian Sch\"afer and Joseph Guinness for helpful comments and discussions.


\section{Proofs \label{app:proofs}}

Before proving the main result (Proposition~\ref{prop:equivalence} in Section~\ref{sec:mathprop}), we shall need the following simple lemma. 
\begin{lemma}\label{lemm:deformation}
Suppose that $y(\cdot) \sim \GP(0,K)$ on $\mathbb{R}^d$ is $q$-reducible with respect to $\psi$. Define $y_0(\cdot) = y \big( \psi^{-1}(\cdot) \big) \sim \GP(0,K_0)$. Then, 
\[
K(\bx , \bx') = K_0 (\| \psi(\bx) - \psi(\bx') \|), \qquad \bx, \bx' \in \mathbb{R}^d.
\]
\end{lemma}

\begin{proof}[Proof of Lemma~\ref{lemm:deformation}]
Note that the isotropic covariance function $K_0$ is only a function of Euclidean distance between inputs. For any inputs $\bx , \bx' \in \mathbb{R}^d$,
\[
K(\bx , \bx') = \cov \left( y(\bx) , y(\bx') \right) = \cov \left( y \big( \psi^{-1}(\psi(\bx)) \big) , y \big( \psi^{-1}(\psi(\bx')) \big) \right) = K_0 \left( \| \psi(\bx) - \psi(\bx') \| \right),
\]
where $\psi(\bx), \psi(\bx') \in \mathbb{R}^q$.
\end{proof}

\begin{proof}[Proof of Proposition~\ref{prop:equivalence}]
From Lemma~\ref{lemm:deformation},
\begin{align*}
\tau_C (i , j) &= \left( 1 - \frac{K(\bx_i , \bx_j)}{\sqrt{K(\bx_i , \bx_i)}\sqrt{K(\bx_j , \bx_j)}} \right)^{1/2}\\
&= \left( 1 - \frac{K_0 \left( \| \psi(\bx_i) - \psi(\bx_j) \| \right)}{\sqrt{K_0 \left( \| \psi(\bx_i) - \psi(\bx_i) \| \right)}\sqrt{K_0 \left( \| \psi(\bx_j) - \psi(\bx_j) \| \right)}} \right)^{1/2}\\
&= \left( 1 - \frac{K_0 \left( \| \psi(\bx_i) - \psi(\bx_j) \| \right)}{K_0 (0)} \right)^{1/2},
\end{align*}
which is strictly increasing in $\tau^\psi_E (i , j) = \| \psi(\bx_i) - \psi(\bx_j) \|$, the Euclidean distances between the corresponding transformations. Then, since each step of the MM ordering only depends on the ranking of distances between inputs, for each $k$,
\begin{equation*}
\argmax_{i \, \in \, \mathcal{I} \setminus \mathcal{I}_{1:k-1}} \,\, \min_{j \, \in \, \mathcal{I}_{1:k-1}} \tau_C(i,j) \;  = \; \argmax_{i \, \in \, \mathcal{I} \setminus \mathcal{I}_{1:k-1}} \,\, \min_{j \, \in \, \mathcal{I}_{1:k-1}} \tau^\psi_E(i,j),
\end{equation*}
and so C-MM of the inputs is identical to E-MM of their transformations. For the same reason, C-NN of the inputs is identical to E-NN of their transformations. Therefore, if the first index is chosen to be same for both C-MM and E-MM orderings, CVecchia of $y(\cdot)$ at the inputs $\bx_1 , \ldots , \bx_n$ is equivalent to EVecchia of $y\big( \psi^{-1}(\cdot) \big)$ at the transformed inputs $\psi(\bx_1) , \ldots , \psi(\bx_n)$.
\end{proof}


\bibliographystyle{apalike}
\bibliography{mendeley,additionalrefs}

\begin{thebibliography}{}

\bibitem[Apanasovich and Genton, 2010]{Apanasovich2010}
Apanasovich, T.~V. and Genton, M.~G. (2010).
\newblock {Cross-covariance functions for multivariate random fields based on
  latent dimensions}.
\newblock {\em Biometrika}, 97(1):15--30.

\bibitem[Banerjee et~al., 2004]{Banerjee2004}
Banerjee, S., Carlin, B.~P., and Gelfand, A.~E. (2004).
\newblock {\em {Hierarchical Modeling and Analysis for Spatial Data}}.
\newblock Chapman {\&} Hall.

\bibitem[Beck, 2017]{beck2017modelling}
Beck, D. (2017).
\newblock Modelling representation noise in emotion analysis using {Gaussian}
  processes.
\newblock In {\em Proceedings of the Eighth International Joint Conference on
  Natural Language Processing (Volume 2: Short Papers)}, pages 140--145.

\bibitem[Beck et~al., 2014]{beck2014joint}
Beck, D., Cohn, T., and Specia, L. (2014).
\newblock Joint emotion analysis via multi-task {Gaussian} processes.
\newblock In {\em Proceedings of the 2014 Conference on Empirical Methods in
  Natural Language Processing (EMNLP)}, pages 1798--1803. ACL.

\bibitem[Cressie and Wikle, 2011]{Cressie2011}
Cressie, N. and Wikle, C.~K. (2011).
\newblock {\em {Statistics for Spatio-Temporal Data}}.
\newblock Wiley, Hoboken, NJ.

\bibitem[Curriero, 2006]{curriero2006use}
Curriero, F.~C. (2006).
\newblock {On the use of non-Euclidean distance measures in geostatistics}.
\newblock {\em Mathematical Geology}, 38(8):907--926.

\bibitem[Datta et~al., 2016a]{Datta2016}
Datta, A., Banerjee, S., Finley, A.~O., and Gelfand, A.~E. (2016a).
\newblock {Hierarchical nearest-neighbor Gaussian process models for large
  geostatistical datasets}.
\newblock {\em Journal of the American Statistical Association},
  111(514):800--812.

\bibitem[Datta et~al., 2016b]{Datta2016a}
Datta, A., Banerjee, S., Finley, A.~O., Hamm, N. A.~S., and Schaap, M. (2016b).
\newblock {Non-separable dynamic nearest-neighbor Gaussian process models for
  large spatio-temporal data with an application to particulate matter
  analysis}.
\newblock {\em Annals of Applied Statistics}, 10(3):1286--1316.

\bibitem[Deriu et~al., 2017]{deriu2017leveraging}
Deriu, J., Lucchi, A., De~Luca, V., Severyn, A., M{\"u}ller, S., Cieliebak, M.,
  Hofmann, T., and Jaggi, M. (2017).
\newblock Leveraging large amounts of weakly supervised data for multi-language
  sentiment classification.
\newblock In {\em Proceedings of the 26th International Conference on World
  Wide Web}, pages 1045--1052.

\bibitem[Eidsvik et~al., 2012]{Eidsvik2012a}
Eidsvik, J., Finley, A.~O., Banerjee, S., and Rue, H. (2012).
\newblock {Approximate Bayesian inference for large spatial datasets using
  predictive process models}.
\newblock {\em Computational Statistics and Data Analysis}, 56(6):1362--1380.

\bibitem[Eidsvik et~al., 2014]{Eidsvik2012}
Eidsvik, J., Shaby, B.~A., Reich, B.~J., Wheeler, M., and Niemi, J. (2014).
\newblock {Estimation and prediction in spatial models with block composite
  likelihoods using parallel computing}.
\newblock {\em Journal of Computational and Graphical Statistics},
  23(2):295--315.

\bibitem[Finley et~al., 2009]{Finley2009}
Finley, A.~O., Sang, H., Banerjee, S., and Gelfand, A.~E. (2009).
\newblock {Improving the performance of predictive process modeling for large
  datasets}.
\newblock {\em Computational Statistics {\&} Data Analysis}, 53(8):2873--2884.

\bibitem[Gao et~al., 2008]{gao2008gaussian}
Gao, P., Honkela, A., Rattray, M., and Lawrence, N.~D. (2008).
\newblock {Gaussian} process modelling of latent chemical species: Applications
  to inferring transcription factor activities.
\newblock {\em Bioinformatics}, 24(16):i70--i75.

\bibitem[Gneiting and Katzfuss, 2014]{Gneiting2014}
Gneiting, T. and Katzfuss, M. (2014).
\newblock {Probabilistic forecasting}.
\newblock {\em Annual Review of Statistics and Its Application}, 1(1):125--151.

\bibitem[Groot et~al., 2011]{groot2011learning}
Groot, P., Birlutiu, A., and Heskes, T. (2011).
\newblock Learning from multiple annotators with {Gaussian} processes.
\newblock In {\em International Conference on Artificial Neural Networks},
  pages 159--164. Springer.

\bibitem[Gu and Wang, 2018]{gu2018scaled}
Gu, M. and Wang, L. (2018).
\newblock Scaled {Gaussian} stochastic process for computer model calibration
  and prediction.
\newblock {\em SIAM/ASA Journal on Uncertainty Quantification},
  6(4):1555--1583.

\bibitem[Guinness, 2018]{Guinness2016a}
Guinness, J. (2018).
\newblock {Permutation and grouping methods for sharpening Gaussian process
  approximations}.
\newblock {\em Technometrics}, 60(4):415--429.

\bibitem[Guinness, 2021]{Guinness2019}
Guinness, J. (2021).
\newblock {Gaussian process learning via Fisher scoring of Vecchia's
  approximation}.
\newblock {\em Statistics and Computing}, 31(25).

\bibitem[Heaton et~al., 2019]{Heaton2017}
Heaton, M.~J., Datta, A., Finley, A.~O., Furrer, R., Guinness, J., Guhaniyogi,
  R., Gerber, F., Gramacy, R.~B., Hammerling, D.~M., Katzfuss, M., Lindgren,
  F., Nychka, D.~W., Sun, F., and Zammit-Mangion, A. (2019).
\newblock {A case study competition among methods for analyzing large spatial
  data}.
\newblock {\em Journal of Agricultural, Biological, and Environmental
  Statistics}, 24(3):398--425.

\bibitem[Johnson and Lindenstrauss, 1984]{johnson1984extensions}
Johnson, W.~B. and Lindenstrauss, J. (1984).
\newblock {Extensions of Lipschitz mappings into a Hilbert space}.
\newblock {\em Contemporary Mathematics}, 26(189-206):1.

\bibitem[Jones et~al., 1998]{Jones1998}
Jones, D.~R., Schonlau, M., and {W. J. Welch} (1998).
\newblock {Efficient global optimization of expensive black-box functions}.
\newblock {\em Journal of Global Optimization}, 13:455--492.

\bibitem[Jones and Zhang, 1997]{Jones1997}
Jones, R.~H. and Zhang, Y. (1997).
\newblock {Models for continuous stationary space-time processes}.
\newblock In Gregoire, T.~G., Brillinger, D.~R., Diggle, P.~J., Russek-Cohen,
  E., Warren, W.~G., and Wolfinger, R.~D., editors, {\em Modelling Longitudinal
  and Spatially Correlated Data}, pages 289--298. Springer, New York.

\bibitem[Karl and Koss, 1984]{karl1984regional}
Karl, T. and Koss, W.~J. (1984).
\newblock Regional and national monthly, seasonal, and annual temperature
  weighted by area, 1895-1983.
\newblock {\em Historical Climatology Series 4--3}, page~38.

\bibitem[Katzfuss, 2017]{Katzfuss2015}
Katzfuss, M. (2017).
\newblock {A multi-resolution approximation for massive spatial datasets}.
\newblock {\em Journal of the American Statistical Association},
  112(517):201--214.

\bibitem[Katzfuss and Gong, 2020]{Katzfuss2017b}
Katzfuss, M. and Gong, W. (2020).
\newblock {A class of multi-resolution approximations for large spatial
  datasets}.
\newblock {\em Statistica Sinica}, 30(4):2203--2226.

\bibitem[Katzfuss and Guinness, 2021]{Katzfuss2017a}
Katzfuss, M. and Guinness, J. (2021).
\newblock {A general framework for Vecchia approximations of Gaussian
  processes}.
\newblock {\em Statistical Science}, 36(1):124--141.

\bibitem[Katzfuss et~al., 2020]{Katzfuss2018}
Katzfuss, M., Guinness, J., Gong, W., and Zilber, D. (2020).
\newblock {Vecchia approximations of Gaussian-process predictions}.
\newblock {\em Journal of Agricultural, Biological, and Environmental
  Statistics}, 25(3):383--414.

\bibitem[Katzfuss et~al., 2022]{Katzfuss2020}
Katzfuss, M., Guinness, J., and Lawrence, E. (2022).
\newblock {Scaled Vecchia approximation for fast computer-model emulation}.
\newblock {\em SIAM/ASA Journal on Uncertainty Quantification}, accepted.

\bibitem[Kennedy and O'Hagan, 2001]{Kennedy2001}
Kennedy, M.~C. and O'Hagan, A. (2001).
\newblock {Bayesian calibration of computer models}.
\newblock {\em Journal of the Royal Statistical Society: Series B},
  63(3):425--464.

\bibitem[Kidd and Katzfuss, 2021]{Kidd2020}
Kidd, B. and Katzfuss, M. (2021).
\newblock {Bayesian nonstationary and nonparametric covariance estimation for
  large spatial data}.
\newblock {\em Bayesian Analysis}, accepted.

\bibitem[Konomi et~al., 2019]{Konomi2019}
Konomi, B.~A., Hanandeh, A.~A., Ma, P., and Kang, E.~L. (2019).
\newblock {Computationally efficient nonstationary nearest-neighbor Gaussian
  process models using data-driven techniques}.
\newblock {\em Environmetrics}, 30(8):e2571.

\bibitem[Liu et~al., 2020]{Liu2018}
Liu, H., Ong, Y.-S., Shen, X., and Cai, J. (2020).
\newblock {When Gaussian process meets big data: A review of scalable GPs}.
\newblock {\em IEEE Transactions on Neural Networks and Learning Systems}.

\bibitem[Maehara, 2013]{maehara2013euclidean}
Maehara, H. (2013).
\newblock Euclidean embeddings of finite metric spaces.
\newblock {\em Discrete Mathematics}, 313(23):2848--2856.

\bibitem[Matousek, 2013]{matousek2013lectures}
Matousek, J. (2013).
\newblock {\em Lectures on Discrete Geometry}, volume 212.
\newblock Springer Science \& Business Media.

\bibitem[Mearns et~al., 2012]{mearns2012north}
Mearns, L.~O., Arritt, R., Biner, S., Bukovsky, M.~S., McGinnis, S., Sain, S.,
  Caya, D., Correia, J., Flory, D., Gutowski, W., et~al. (2012).
\newblock The {North} {American} regional climate change assessment program:
  Overview of phase {I} results.
\newblock {\em Bulletin of the American Meteorological Society},
  93(9):1337--1362.

\bibitem[Mearns et~al., 2009]{mearns2009regional}
Mearns, L.~O., Gutowski, W., Jones, R., Leung, R., McGinnis, S., Nunes, A., and
  Qian, Y. (2009).
\newblock A regional climate change assessment program for {North} {America}.
\newblock {\em Eos, Transactions American Geophysical Union}, 90(36):311--311.

\bibitem[Messier and Katzfuss, 2021]{Messier2020}
Messier, K.~P. and Katzfuss, M. (2021).
\newblock {Scalable penalized spatiotemporal land-use regression for
  ground-level nitrogen dioxide}.
\newblock {\em Annals of Applied Statistics}, 15(2):688--710.

\bibitem[Min et~al., 2021]{min2021recent}
Min, B., Ross, H., Sulem, E., Veyseh, A. P.~B., Nguyen, T.~H., Sainz, O.,
  Agirre, E., Heinz, I., and Roth, D. (2021).
\newblock Recent advances in natural language processing via large pre-trained
  language models: A survey.
\newblock {\em arXiv:2111.01243}.

\bibitem[Paciorek and Schervish, 2006]{Paciorek2006}
Paciorek, C. and Schervish, M. (2006).
\newblock {Spatial modelling using a new class of nonstationary covariance
  functions}.
\newblock {\em Environmetrics}, 17(5):483–506.

\bibitem[Perrin and Meiring, 2003]{perrin2003nonstationarity}
Perrin, O. and Meiring, W. (2003).
\newblock Nonstationarity in $\mathbb{R}^n$ is second-order stationarity in
  $\mathbb{R}^{2n}$.
\newblock {\em Journal of Applied Probability}, 40(3):815--820.

\bibitem[Perrin and Monestiez, 1999]{perrin1999modelling}
Perrin, O. and Monestiez, P. (1999).
\newblock Modelling of non-stationary spatial structure using parametric radial
  basis deformations.
\newblock In {\em geoENV II—Geostatistics for Environmental Applications},
  pages 175--186. Springer.

\bibitem[Perrin and Schlather, 2007]{perrin2007can}
Perrin, O. and Schlather, M. (2007).
\newblock Can any multivariate {Gaussian} vector be interpreted as a sample
  from a stationary random process?
\newblock {\em Statistics \& Probability Letters}, 77(9):881--884.

\bibitem[Perrin and Senoussi, 2000]{perrin2000reducing}
Perrin, O. and Senoussi, R. (2000).
\newblock Reducing non-stationary random fields to stationarity and isotropy
  using a space deformation.
\newblock {\em Statistics \& Probability Letters}, 48(1):23--32.

\bibitem[Porcu et~al., 2010]{porcu2010non}
Porcu, E., Matkowski, J., and Mateu, J. (2010).
\newblock On the non-reducibility of non-stationary correlation functions to
  stationary ones under a class of mean-operator transformations.
\newblock {\em Stochastic Environmental Research and Risk Assessment},
  24(5):599--610.

\bibitem[Rasmussen and Williams, 2006]{Rasmussen2006}
Rasmussen, C.~E. and Williams, C. K.~I. (2006).
\newblock {\em {Gaussian Processes for Machine Learning}}.
\newblock MIT Press.

\bibitem[Risser and Turek, 2020]{Risser2020}
Risser, M.~D. and Turek, D. (2020).
\newblock {Bayesian inference for high-dimensional nonstationary Gaussian
  processes}.
\newblock {\em Journal of Statistical Computation and Simulation}.

\bibitem[Sacks et~al., 1989]{Sacks1989}
Sacks, J., Welch, W., Mitchell, T., and Wynn, H. (1989).
\newblock {Design and analysis of computer experiments}.
\newblock {\em Statistical Science}, 4(4):409--435.

\bibitem[Sang et~al., 2011]{Sang2011a}
Sang, H., Jun, M., and Huang, J.~Z. (2011).
\newblock {Covariance approximation for large multivariate spatial datasets
  with an application to multiple climate model errors}.
\newblock {\em Annals of Applied Statistics}, 5(4):2519--2548.

\bibitem[Sch{\"{a}}fer et~al., 2021]{Schafer2020}
Sch{\"{a}}fer, F., Katzfuss, M., and Owhadi, H. (2021).
\newblock {Sparse Cholesky factorization by Kullback-Leibler minimization}.
\newblock {\em SIAM Journal on Scientific Computing}, 43(3):A2019--A2046.

\bibitem[Schmidt and O'Hagan, 2003]{Schmidt2003}
Schmidt, A.~M. and O'Hagan, A. (2003).
\newblock {Bayesian inference for non-stationary spatial covariance structure
  via spatial deformations}.
\newblock {\em Journal of the Royal Statistical Society. Series B},
  65(3):743--758.

\bibitem[Snelson and Ghahramani, 2007]{Snelson2007}
Snelson, E. and Ghahramani, Z. (2007).
\newblock {Local and global sparse Gaussian process approximations}.
\newblock In {\em Artificial Intelligence and Statistics 11 (AISTATS)}.

\bibitem[Stein, 2005]{Stein2005}
Stein, M.~L. (2005).
\newblock {Nonstationary spatial covariance functions}.
\newblock {\em Technical Report No. 21, University of Chicago}.

\bibitem[Stein et~al., 2004]{Stein2004}
Stein, M.~L., Chi, Z., and Welty, L. (2004).
\newblock {Approximating likelihoods for large spatial data sets}.
\newblock {\em Journal of the Royal Statistical Society: Series B},
  66(2):275--296.

\bibitem[Sun and Stein, 2016]{Sun2016}
Sun, Y. and Stein, M.~L. (2016).
\newblock {Statistically and computationally efficient estimating equations for
  large spatial datasets}.
\newblock {\em Journal of Computational and Graphical Statistics},
  25(1):187--208.

\bibitem[Van~Dongen and Enright, 2012]{van2012metric}
Van~Dongen, S. and Enright, A.~J. (2012).
\newblock Metric distances derived from cosine similarity and {Pearson} and
  {Spearman} correlations.
\newblock {\em arXiv:1208.3145}.

\bibitem[Varin, 2008]{Varin2008}
Varin, C. (2008).
\newblock {On composite marginal likelihoods}.
\newblock {\em AStA Advances in Statistical Analysis}, 92(1):1--28.

\bibitem[Vecchia, 1988]{Vecchia1988}
Vecchia, A. (1988).
\newblock {Estimation and model identification for continuous spatial
  processes}.
\newblock {\em Journal of the Royal Statistical Society, Series B},
  50(2):297--312.

\bibitem[Vu et~al., 2020]{vu2020modeling}
Vu, Q., Zammit-Mangion, A., and Cressie, N. (2020).
\newblock Modeling nonstationary and asymmetric multivariate spatial
  covariances via deformations.
\newblock {\em arXiv:2004.08724}.

\bibitem[White and Porcu, 2019]{white2019nonseparable}
White, P. and Porcu, E. (2019).
\newblock {Nonseparable covariance models on circles cross time: A study of
  Mexico City ozone}.
\newblock {\em Environmetrics}, page e2558.

\bibitem[Williams et~al., 2007]{williams2007multi}
Williams, C., Bonilla, E.~V., and Chai, K.~M. (2007).
\newblock Multi-task {Gaussian} process prediction.
\newblock {\em Advances in Neural Information Processing Systems}, pages
  153--160.

\bibitem[Witsenhausen, 1986]{witsenhausen1986minimum}
Witsenhausen, H.~S. (1986).
\newblock Minimum dimension embedding of finite metric spaces.
\newblock {\em Journal of Combinatorial Theory, Series A}, 42(2):184--199.

\bibitem[Yang, 2019]{Yang2019}
Yang, G. (2019).
\newblock {Wide feedforward or recurrent neural networks of any architecture
  are Gaussian processes}.
\newblock In {\em Advances in Neural Information Processing Systems}, pages
  9951--9960.

\bibitem[Yu et~al., 2017]{yu2017geometry}
Yu, C.~D., Levitt, J., Reiz, S., and Biros, G. (2017).
\newblock {Geometry-oblivious FMM for compressing dense SPD matrices}.
\newblock In {\em Proceedings of the International Conference for High
  Performance Computing, Networking, Storage and Analysis}, page~53. ACM.

\bibitem[Zhang et~al., 2021]{zhang2021high}
Zhang, L., Banerjee, S., and Finley, A.~O. (2021).
\newblock {High-dimensional multivariate geostatistics: A Bayesian
  matrix-normal approach}.
\newblock {\em Environmetrics}, 32(4):e2675.

\end{thebibliography}


\end{document}